\documentclass{article}

\usepackage{float}
\usepackage{xcolor}
\usepackage{amssymb,latexsym,amsmath,bbm,graphicx,amsthm}
\usepackage{bussproofs}

\usepackage[all]{xy}

\usepackage{amsfonts}

\usepackage{enumerate}

\newtheorem{theo}{Theorem}

\newtheorem{prop}{Proposition}
\newtheorem{lemma}{Lemma}

\theoremstyle{definition}
\newtheorem{defi}{Definition}
\theoremstyle{definition}

\theoremstyle{definition}

\theoremstyle{definition}

\theoremstyle{definition}
\newtheorem{conv}{Convention}

\theoremstyle{definition}

\newenvironment{pfclaim}{\begin{trivlist}\item[]{\sc Proof of
Claim}}{\hfill {\mbox{$\blacktriangleleft$}}\end{trivlist}}
\newtheorem{claim2}{\sc Claim}
\newenvironment{claim}{\begin{claim2}\rm}{\end{claim2}\rm}

\newcommand{\psf}{\mathcal{P}}

\newcommand{\at}[1]{\mathsf{#1}\!:\!}
\newcommand{\ats}[2]{\mathsf{#1}_{#2}\!:\!}

\newcommand{\bfi}{\mathsf{i}}
\newcommand{\bfj}{\mathsf{j}}
\newcommand{\bfk}{\mathsf{k}}

\newcommand{\plrsat}{{\textbf{Fal}}}
\newcommand{\plrval}{{\textbf{Ver}}}
\newcommand{\tset}[1]{[\![#1]\!]}

\newcommand{\sysinf}{{\textbf{Inf}}}
\newcommand{\sysann}{{\textbf{Saf}}}

\newcommand{\napprox}{\not\approx}
\newcommand{\mpt}{^\varepsilon}
\newcommand{\an}[1]{^{\mathsf{#1}}}
\newcommand{\ans}[2]{^{\mathsf{#1}_{#2}}}
\newcommand{\crl}[1]{\mathsf{#1} \vdash}
\newcommand{\crm}{\varepsilon \vdash}
\newcommand{\sfa}{\mathsf{a}}
\newcommand{\sfb}{\mathsf{b}}
\newcommand{\sfc}{\mathsf{c}}
\newcommand{\sfd}{\mathsf{d}}
\newcommand{\sfx}{\mathsf{x}}
\newcommand{\sfy}{\mathsf{y}}
\newcommand{\sfz}{\mathsf{z}}

\newcommand{\sitail}[1]{\Uparrow\!\!(\sigma,#1)}
\newcommand{\rst}[1]{\vert_{#1}}
\newcommand{\model}{\mathcal{M}}
\newcommand{\uls}[1]{\underline{#1}}

\title{A circular proof system for the hybrid $\mu$-calculus}
\author{Sebastian Enqvist \\ Department of Philosophy,
Stockholm University}

\begin{document}
\maketitle

\begin{abstract}We present a circular and cut-free proof system for the hybrid $\mu$-calculus and prove its soundness and completeness. The system uses names for fixpoint unfoldings, like the circular proof system for the $\mu$-calculus previously developed by Stirling.
\end{abstract}

\section{Introduction}

Circular and non-wellfounded proofs are a powerful method for reasoning with fixpoints, and have been considered in a number of contexts  \cite{sant:calc02, dam:mu02, dam:struc03, broth:comp07, broth:sequ10, sham:circ14}. For the modal $\mu$-calculus, a circular proof system with names for keeping track of fixpoint unfoldings was developed by Stirling \cite{stir:tabl14}, building on work by Jungteerapanich \cite{jung:tabl10} and bearing similarities with earlier systems using variables for ordinal approximations \cite{dam:mu02}. Recently Stirling's system has been simplified and used by Afshari and Leigh to give a cut-free complete sequent system for the modal $\mu$-calculus \cite{afsh:cut17}. This provides a novel completeness proof for Kozen's axiomatization \cite{koze:resu83} that avoids the intricate detour via disjunctive normal forms in Walukiewicz's proof \cite{walu:comp01}. 

There are two parallell motivations to continue research into circular proof systems for fixpoint logics. First, they are interesting in their own right from a proof theoretic perspective, providing tools for deeper proof-theoretic analysis of fixpoint logics. Second, they  provide a promising framework for proving novel  completeness results and providing proof systems for fixpoint logics where a complete axiomatization is currently lacking. This is witnessed by recent work where a circular proof system was developed for Parikh's dynamic logic of games \cite{enqv:comp19}, and used to settle the open problem of completeness of Parikh's original axiom system \cite{pari:logi85}.
 
The present work is an attempt to take a step towards exploring the use of circular proofs to provide complete finitary proof systems for richer \emph{extensions} of the modal $\mu$-calculus. A number of such extensions have been presented in the literature, including the two-way or ``full'' $\mu$-calculus \cite{vard:reas98}, hybrid $\mu$-calculus \cite{satt:hybr01} and guarded fixpoint logic \cite{grad:guar99}. In many cases such extensions remain decidable, without any increase in complexity. However, complete proof systems mostly appear to be lacking. Some work in this area does exist:  a generic completeness result for coalgebraic versions of the $\mu$-calculus (including extensions like the graded $\mu$-calculus) was presented in \cite{enqv:comp19a}, and an infinitary proof system for the two-way $\mu$-calculus was proved complete in \cite{afsh:infi19}.  Our hope is that circular proofs can be developed further as a method to provide complete proof systems for such expressive extensions of the $\mu$-calculus. 

As a proof of concept, we shall develop a cut-free Stirling-style circular proof system for a version of the hybrid $\mu$-calculus. The language we consider is a relatively gentle, but still interesting, extension of the modal $\mu$-calculus. It adds two features of hybrid logic to the language: \emph{nominals}, which are used to name points in a model, and \emph{satisfaction operators} that describe what is true at a named point in a model. This logic was studied by Tamura in \cite{tamu:smal13}, where it was shown that it has the finite model property. This is in contrast with the two-way $\mu$-calculus, and consequently with Sattler and Vardi's original version of the hybrid $\mu$-calculus which includes backward modalities. We hope that the techniques developed here can be extended to provide complete proof systems for the hybrid $\mu$-calculus with backward modalities as well, and perhaps even eventually for guarded fixpoint logic. But already the introduction of nominals and satisfaction operators presents some non-trivial challenges, and addressing these difficulties gives some guidelines on how to deal with proof theory for fixpoint logics that lack the tree model property. In a manner of speaking, we are continuing here along Sattler and Vardi's line of working with logics that lack the tree model property ``as if they had the tree model property'' \cite{satt:hybr01}, but taking the idea in a proof-theoretic direction.

\section{Preliminaries}

\subsection{The hybrid $\mu$-calculus}
 
The hybrid $\mu$-calculus was initially introduced by Sattler and Vardi in \cite{satt:hybr01}. Their version of the language included a global modality and converse modalities. Here, we shall be considering the weaker version of the hybrid $\mu$-calculus that was studied by Tamura in \cite{tamu:smal13}. For ease of notation we consider the language with only a single box and diamond, but all the results and proofs presented here easily extend to a multi-modal version of the language. 

The language $\mathcal{L}$ of the hybrid $\mu$-calculus is given by the following grammar:
$$\varphi := p \mid \neg p \mid \bfi  \mid \neg \bfi \mid \varphi \vee \varphi \mid \varphi \wedge \varphi  \mid \Diamond \varphi \mid \Box \varphi \mid \at{i} \varphi \mid \mu x. \varphi \mid \nu x. \varphi $$  

Here, $p$ and $x$ are members of a fixed countably infinite supply $\mathsf{Prop}$ of propositional variables, and $\bfi$ comes from a fixed countably infinite supply $\mathsf{Nom}$ of nominals.  For $\eta x. \varphi$ with $\eta \in \{\mu,\nu\}$, we impose the usual constraint that no occurrence of $x$ in $\varphi$ is in the scope of a negation, and we also require that each occurrence of $x$ in $\varphi$ is within the scope of some modality ($\Box$ or  $\Diamond$). This latter extra constraint means that we restrict attention to \emph{guarded} formulas. This is a fairly common assumption,  and it is well known that removing the constraint of guardedness does not increase the expressive power of the language. It is not an entirely innocent assumption however, since putting a formula in its guarded normal form may cause an exponential blow-up in the size of a formula \cite{brus:guar15}.   Note also that the language is presented in negation normal form. It is routine to verify, given the semantics presented below, that the language is semantically closed under negation, and furthermore there is a simple effective procedure for converting formulas in the extended language with explicit negation of all formulas into formulas in negation normal form.

Free and bound variables of a formula are defined in the usual manner. A \emph{literal} is a formula of the form $p$ or $\neg p$ where $p \in \mathsf{Prop}$, or of the form $\bfi$ or $\neg \bfi$ where $\bfi \in \mathsf{Nom}$. We introduce the following abbreviations:
$$\bfi \approx \bfj := \at{i} \bfj \quad \quad \quad \bfi \napprox \bfj := \at{i}\neg \bfj$$
These formulas express identity and non-identity, respectively, of the values assigned to the nominals $\bfi,\bfj$ in a model. 

\begin{defi}
Let $\varphi$ be any formula in $\mathcal{L}$ and let $x,y \in \mathsf{Prop}$ be bound variables in $\varphi$. We say that $y$ is \emph{dependent} on $x$, written $x <_\varphi y$, if there is a subformula of $\varphi$ of the form $\eta y.\psi$ in which there is a free occurrence of $x$.  We denote the reflexive closure of $<_\varphi$ by $\leq_\varphi$. 
\end{defi}

\begin{defi}
We say that a formula $\varphi$ is \emph{locally well-named} if $<_\varphi$ is irreflexive, no variable occurs both free and bound in $\varphi$, and no variable is bound by both $\mu$ and $\nu$ in $\varphi$.  
\end{defi} 
Note that every formula is equivalent to a locally well-named one up to renaming of bound variables ($\alpha$-equivalence). 
\begin{prop}[Afshari \& Leigh -17]
If $\eta x. \varphi(x)$ is locally well-named then so is $\varphi(\eta x. \varphi)$.
\end{prop}
\begin{conv}
We shall assume throughout the paper that all formulas are locally well-named. Given a locally well-named formula we refer to a bound variable $x$ as a $\mu$-variable if it is bound (only) by $\mu$ in $\varphi$, and a $\nu$-variable if it is bound (only) by $\nu$.   
\end{conv}
Semantics of the hybrid $\mu$-calculus is a simple extension of the usual Kripke semantics for the modal $\mu$-calculus.
\begin{defi}
A \emph{Kripke model} is a tuple $\model = (W,R,V,A)$ where $W$ is a non-empty set members of which will be referred to as \emph{points}, $R \subseteq W \times W$ is the \emph{accessibility relation} over $W$, $V : \mathsf{Prop} \to \psf(W)$ is a \emph{valuation} of the propositional variables and $A : \mathsf{Nom} \to W$ is an \emph{assignment} of a value in $W$ to each nominal. 
\end{defi}
Given a Kripke model $\model = (W,R,V,A)$, the interpretation $\tset{\varphi}_\model$ of a formula $\varphi$ is defined by the usual recursive clauses for boolean connectives and modalities. Semantics of least fixpoint operators is given according to the Knaster-Tarski Theorem \cite{knas:theo28, tars:latt55} as:
$$\tset{\mu x. \varphi(x)}_\model := \bigcap\{Z \subseteq W \mid \tset{\varphi}_{\model[Z/x]} \subseteq Z\},$$
 where $\model[Z/x]$ is like $\model$ except that its valuation maps the variable $x$ to $Z$. For greatest fixpoint operators we have the dual definition:
$$\tset{\nu x. \varphi(x)}_\model := \bigcup\{Z \subseteq W \mid Z \subseteq \tset{\varphi}_{\model[Z/x]}\}.$$
For nominals and satisfaction operators, we have the following clauses: $$\tset{\bfi}_\model = \{A(\bfi)\},$$ and $$\tset{\at{i}\varphi}_\model = \{w \in W \mid A(\bfi) \in \tset{\varphi}\}.$$ In other words, $\tset{\at{i} \varphi}_\model = W$ if $A(\bfi) \in \tset{\varphi}$, and $\tset{\at{i} \varphi}_\model = \emptyset$ otherwise. Given a formula $\varphi$ and a pointed Kripke model $(\model,w)$ (a model with a distinguished point), we write $\model, w \Vdash \varphi$ to say that $w \in \tset{\varphi}_\model$. 

This semantics may be referred to as the \emph{denotational} semantics of the $\mu$-calculus, or as (a special case of) the \emph{algebraic} semantics. The $\mu$-calculus also has an \emph{operational} semantics in the form of a game semantics, which is often easier to work with and neatly captures the intuitive meaning of least and greatest fixpoints (i.e. ``finite looping'' vs ``infinite looping''). 
\begin{defi}
Let $\Gamma$ be a set of $\mathcal{L}$-formulas. We say that $\Gamma$ is \emph{Fischer-Ladner closed} if the following conditions hold:
\begin{itemize}
\item If $\varphi O \psi \in \Gamma$  where $O \in \{\wedge, \vee \}$ then $\{\varphi,\psi\} \subseteq \Gamma$.
\item If $O\varphi \in \Gamma$ where $O \in \{\Diamond, \Box\} $ then $\varphi \in \Gamma$.
\item If $\at{i} \varphi \in \Gamma$ where $\bfi \in \mathsf{Nom}$ then $\varphi \in \Gamma$.
\item If $\eta x.\varphi(x) \in \Gamma$ where $\eta \in \{\mu, \nu\}$ then $\varphi(\eta x.\varphi(x))\in \Gamma$.   
\end{itemize}
The \emph{Fischer-Ladner closure} $c(\varphi)$ of a formula $\varphi$ is the smallest Fischer-Ladner closed set of formulas containing $\varphi$. 
\end{defi}
The following result is a well known fact:
\begin{prop}
The Fischer-Ladner closure of any formula is finite, and its size is linear in the length of the formula.  
\end{prop}

Throughout the paper we assume familiarity with basic notions concerning board games and parity games. Given a Kripke model $\model = (W,R,V,A)$, the \emph{evaluation game} for a formula $\rho \in \mathcal{L}$ in the model $\model$ is a two-player board game between players $\plrval,\plrsat$, the set of positions of which is $W \times c(\rho)$, with player assignments and moves defined as follows:
\begin{itemize}
\item For a position of the form $(w,l)$ where $l$ is a literal, the set of availabe moves is $\emptyset$. The position is assigned to $\plrsat$ if $\model,w \Vdash l$ and is assigned to $\plrval$ otherwise. 
\item For a position of the form  $(w,\varphi O \psi)$ where $ O \in \{\wedge, \vee \}$, the available moves are $(w,\varphi)$ and $(w,\psi)$. The position is assigned to $\plrval $ if $O = \vee$ and is assigned to $\plrsat$ if $O = \wedge$.  
\item For a position of the form $(w,O\varphi)$ where $O \in \{\Diamond,\Box\}$, the set of available moves is $\{v \in W \mid w R v\}$. The position is assigned to $\plrval$ if $O = \Diamond$ and is assigned to $\plrsat$ if $O = \Box$. 
\item For a position of the form $(w,\at{i}\varphi)$, the unique avaliable move is $(A(\bfi),\varphi)$. The player assignment is arbitrary in this case since there is only one move, but as a convention we assign such positions to player $\plrval$. 
\item For a position of the form $(w,\eta x.\varphi(x))$, the unique available move is $(w,\varphi(\eta x.\varphi))$. By convention we assign such positions to \plrval. 
\end{itemize}
Partial plays, full plays and strategies for players are defined as usual. Note that if a full play is finite, then the player to which the last position is assigned must be ``stuck'', i.e. the set of available moves is empty. So the winning condition of finite full plays is defined by declaring the player who got stuck to be the loser of the play. For infinite plays $(w_1,\varphi_1)(w_2,\varphi_2)(w_3,\varphi_3)...$, say that a fixpoint variable $x$ is \emph{unfolded} at the index $i$ if $\varphi_i$ is of the form $\eta x. \psi(x)$.
\begin{prop}
For any (locally well-named) formula $\rho$ and any infinite play $\pi$ in the evaluation game in $\model$, there is a unique $<_\rho$-minimal variable $x$ that is unfolded infinitely many times on $\pi$. 
\end{prop}  
We shall often refer to the $<_\rho$-minimal variable unfolded infinitely often on $\pi$ as the \emph{highest ranking} variable that is unfolded infinitely often. We can now define the winning condition of infinite plays: the winner is \plrval{} if the highest ranking variable that gets unfolded infinitely often is a $\nu$-variable (relative to $\rho$), and the winner is $\plrsat$ otherwise. 

Strategies and winning strategies of players are defined as usual. A strategy is called \emph{positional} if it only depends on the last position of a play, i.e. it can be described as a choice function from positions to available moves. Since the evaluation game is  a parity game, and parity games have positional determinacy \cite{emer:tree91, ziel:infi98}, we have:
\begin{prop}
\label{p:parity-games}
The evaluation game of any formula in a model is determinate, and the winning player at any given position has a positional winning strategy.
\end{prop}
As expected the operational semantics agrees with the denotational one:
\begin{prop}
Given a pointed Kripke model $(\model, w)$ and a formula $\rho$, we have $\model, w \Vdash \rho$ if and only if the position $(w,\rho)$ is winning for \plrval{} in the evaluation game. 
\end{prop}

\subsection{Trees and tree languages}

We will need some basic concepts concerning trees and tree languages. Given two words $\sfa,\sfb$ over some given alphabet we will use the notation $\sfa \cdot \sfb$ for the operation of concatenation. 
\begin{defi}
A \emph{tree} $T$ is a subset of $\mathbb{N}^*$, i.e. a non-empty set of words over the set of natural numbers, that is closed under prefixes and such that whenever $u \cdot m \in T$ and $k < m$, $u \cdot k \in T$. The empty word is called the \emph{root} of the tree. A \emph{ranked alphabet} is a set $\Sigma$ together with a map $\mathsf{ar} : \Sigma \to \omega$ assigning an arity to each member of $\Sigma$. Given a ranked alphabet $\Sigma$, a \emph{$\Sigma$-labelled tree} is a tree $T$ together with a mapping $l : T \to \Sigma$, satisfying the following constraint: if $u \in T$, then the children of $u$ in $T$ are: $$u \cdot 0, ..., u\cdot (\mathsf{ar}(l(\vec{n})) - 1).$$ Finally, a \emph{tree language} over some alphabet $\Sigma$ is a set of $\Sigma$-labelled trees. 
\end{defi}

\begin{defi}
Let $T$ be a $\Sigma$-labelled tree where $\Sigma$ is a ranked alphabet and $l$ is the labelling function. A $\Sigma$-labelled tree $T'$ with labelling function $l'$ is said to be a \emph{subtree} of $T$ if there is some $u \in T$ such that:
\begin{itemize}
\item $T' = \{v \in \mathbb{N}^* \mid u \cdot v \in T\}$
\item $l'(v) = l(u \cdot v)$ for all $v \in T'$.
\end{itemize} 
Given $u \in T$ we call $T'$ the (generated) \emph{subtree rooted at} $u$, denoted $T\rst{u}$. A labelled tree $T$ is said to be \emph{regular} if it has only finitely many subtrees. 
\end{defi}

The \emph{monadic second-order language} for $\Sigma$-labelled trees has the signature consisting of unary predicates corresponding to labels in $\Sigma$ and binary predicates for $i$-th successor relation corresponding to each $i \in \mathsf{max}\{\mathsf{ar}(\sigma) \mid \sigma \in \Sigma\}$. A tree language is said to be MSO-definable if there is a formula in the monadic second-order language for the corresponding alphabet that is satisfied by precisely those $\Sigma$-labelled trees that belong to the tree language.  The following is a slight reformulation of a well-known result in the theory of automata on infinite trees:
\begin{theo}[Rabin's Basis Theorem]
Every non-empty MSO-definable tree language contains a regular tree.
\end{theo}

\begin{defi}
A ($\Sigma$-labelled) \emph{tree with back-edges} is a $\Sigma$-labelled tree $T$ together with a partial map $f : T \to T$ such that every member of $\mathsf{dom}(f)$ is a leaf of $T$, and $f(u)$ is a proper prefix of $u$ for each $u\in \mathsf{dom}(f)$.  
\end{defi}

\begin{defi}
Let $(T,f)$ be a finite tree with back-edges and let $u \in T$. The \emph{$f$-unfolding} of $T$ at $u$, denoted $\mathsf{unf}(T,f,u)$, is the infinite tree obtained by the following coinductive definition: 
$$\mathsf{unf}(T,f,u) := T\rst{u}[\mathsf{unf}(T,f,f(l))/l \mid l \in \mathsf{dom}(f)]$$
where  the operation of substituting a tree for a leaf is defined as expected.  We write $\mathsf{unf}(T,f)$ as short-hand for $\mathsf{unf}(T,f,\varepsilon)$. 
\end{defi}
The previous definition could be replaced with a more standard inductive construction, which would be more complicated and less direct.  The reader who feels uncomfortable with such informal usage of coinduction might consult \cite{koze:prac17} for reassurance.  

\begin{prop}[Folklore]
\label{p:folk}
A tree is regular iff it is the unfolding of some finite tree with back edges. 
\end{prop}

\section{Infinite proofs}

In this section we define an infinite sequent-style proof system $\sysinf$ for the hybrid $\mu$-calculus. This proof system will be used as a tool to prove completeness of the finite circular proof system that will be introduced in Section \ref{s:sysann}. The infinite system presented here is essentially dual to an infinite tableau system for the hybrid $\mu$-calculus. An important difference from the tableaux developed by Sattler and Vardi in \cite{satt:hybr01} is that the system is cut-free, which is required since the finitary circular system we shall present later will also be cut-free. Sattler and Vardi's approach relies on ``guessing'' all the relevant information about some nominals at the start of the tableau construction. In the dual setting of sequent calculi this amounts to starting the proof construction with a series of cuts.  

\subsection{The system $\sysinf$}
We will work with a sequent style proof system, where a sequent is a finite set of formulas interpreted as an implicit disjunction.  
It will be convenient to require that every formula in a sequent starts with some satisfaction operator, so each sequent has the form:
$$\ats{i}{1} \varphi_1 ,...,\ats{i}{n} \varphi_n$$
This is without loss of generality, since an \sysinf-proof for a formula $\varphi$ can be defined as a proof for the sequent $\{\at{i}\varphi\}$ where $\bfi$ is some fresh nominal not appearing in $\varphi$. Clearly $\at{i}\varphi$ is then valid if and only if $\varphi$ is. Sequents will be treated as plain sets rather than multi-sets, so we do not require contraction as a structural rule. 

The system has two axioms, which are the law of exluded middle and an identity axiom: 
$$\at{i}p, \at{i}\neg p \quad \quad \quad \bfi \approx \bfi$$
Here, $p$ is a nominal or a propositional variable. Rules of inference are given in Figure \ref{fig:inf-rules}. We remark that, in the modal rule $\mathsf{Mod}$, the nominal $\bfj$ must be fresh, i.e. it cannot appear in any formula in the conclusion of the rule. 
\begin{figure}[h]
\fbox{
\begin{minipage}[t]{.33\textwidth}
\begin{prooftree}
\AxiomC{$\Gamma, \at{i}\varphi \wedge \psi, \at{i}\varphi$}
\AxiomC{$\Gamma, \at{i}\varphi \wedge \psi, \at{i} \psi$}
\RightLabel{$\wedge$}
\BinaryInfC{$\Gamma, \at{i}\varphi \wedge \psi$}
\end{prooftree}
\vspace{.5cm}
\begin{prooftree}
\AxiomC{$\Gamma, \at{i}\phi, \at{j} \phi, \bfi \napprox \bfj$}
\RightLabel{$\mathsf{Eq}$}
\UnaryInfC{$\Gamma, \at{i}\phi, \bfi \napprox \bfj$}
\end{prooftree}

\begin{prooftree}
\AxiomC{$\Gamma,\at{i}\Box \varphi, \at{i}\Diamond \Psi,\at{j}\varphi,\at{j}\Psi$}
\RightLabel{$\mathsf{Mod}$}
\UnaryInfC{$\Gamma, \at{i}\Box \varphi, \at{i}\Diamond \Psi$}
\end{prooftree}
\end{minipage}
\begin{minipage}[t]{.33\textwidth}
\vspace{1cm}
\begin{prooftree}
\AxiomC{$\Gamma, \at{i}\eta x. \phi(x), \at{i}\phi(\eta x. \phi(x))$} 
\RightLabel{$\eta x$}
\UnaryInfC{$\Gamma, \at{i}\eta x. \phi(x)$}
\end{prooftree}

\begin{prooftree}
 \AxiomC{$\Gamma, \at{i}(\at{j}\varphi), \at{j}\varphi$}
\RightLabel{$\mathsf{Glob}$}
\UnaryInfC{$\Gamma, \at{i}(\at{j}\varphi)$}
\end{prooftree}
\end{minipage}
\begin{minipage}[t]{.33\textwidth}
\begin{prooftree}
\AxiomC{$\Gamma, \at{i}\varphi \vee \psi, \at{i}\varphi,\at{i} \psi$}
\RightLabel{$\vee$}
\UnaryInfC{$\Gamma, \at{i}\varphi \vee \psi $}
\end{prooftree}
\vspace{.5cm}
\begin{prooftree}
\AxiomC{$\Gamma,  \bfi \napprox \bfj, \bfj \napprox \bfi$}
\RightLabel{$\mathsf{Com}$}
\UnaryInfC{$\Gamma,  \bfi \napprox \bfj$}
\end{prooftree}

\begin{prooftree}
\AxiomC{$\Gamma$}
\RightLabel{$\mathsf{Weak}$}
\UnaryInfC{$\Gamma \cup \Psi$}
\end{prooftree}
\end{minipage}
}
\caption{Rules}
\label{fig:inf-rules}
\end{figure}

The expression $\at{i}\Diamond \Psi$ is short-hand for $\{\at{i}\Diamond \psi \mid \psi \in \Psi\}$, and likewise $\at{j}\Psi$ abbreviates $\{\at{j}\psi \mid \psi \in \Psi\}$. The intuition behind the modal rule is that, if the formulas $\Box \varphi, \Diamond \psi_1,...,\Diamond \psi_n$ are all false at a point named $\bfi$, then this must be witnessed by some point that we can give an arbitrary name $\bfj$, and at which all the formulas $\varphi,\psi_1,...,\psi_n$ are false.
It will be useful to think of proofs as being constructed from the root upwards. Note that the rules allow the principal formulas to appear in the premise of a rule, rather than being discarded. This is to be expected, since we are attaching information to nominals, and information that has been established about a nominal at one point in a proof may be needed later. In particular, we do allow that the premises and conclusion of a rule application are the same sequent.

\begin{defi}
A rule application is said to be \emph{repeating} if all premises are equal to the conclusion. 
\end{defi} 

Of course we need to be careful not to let the number of nominals appearing in a sequent grow unboundedly, since we want to construct finite proofs in the end. The weakening rule can be applied for this purpose, and it needs to be applied strategically to remove formulas that we can be sure will not be needed anymore, without losing information that may be needed later. 

In an application of the modal rule as shown in Figure \ref{fig:inf-rules}, we refer to $\at{i}\Box \varphi $ as the principal formula. In an application of the rule $\mathsf{Eq}$ as shown in the figure, the formula $\at{i}\phi$ is  called the principal formula and $\bfi \napprox \bfj$ the \emph{side formula}. In all other cases where a notion of principal formula makes sense, it should be clear from the form of the rules what the principal formula is.

A \emph{\sysinf-proof}, or proof-tree, is a $\Sigma$-labelled tree where the members of $\Sigma$ specify the sequent appearing at a node, the rule application of which the node is the conclusion (if any), and which formula the rule was applied to, and such that the labels of children of a node are the premises of the specified rule application.   We will not define the alphabet more precisely than this, but trust that this informal description will  be sufficiently clear. Ranks will be determined by the number of premises of  rule applications, so that a label specifying an application of the $\wedge$-rule for example will have rank $2$.  We shall often abuse terminology slightly by referring to the sequent appearing at a node in a proof as the label of the node. To distinguish \emph{valid} proofs from invalid ones, we need a notion of trace. 

\begin{defi}
A \emph{partial trace} $\vec{t}$ (of length $k \leq \omega$) on a branch $\beta$ of an \sysinf-proof $\Pi$ is a sequence $(u_j,\ats{i}{j}\psi_j)_{j < k}$ such that for each $j$, $u_j$ is a node on $\beta$ whose label contains $\ats{i}{j}\psi_j$, 
$u_{j + 1}$ is the unique child of $u_j$ in $\beta$ whenever $j + 1 < k$, and one of the following conditions holds if $j + 1 < k$:
\begin{enumerate}
\item $\ats{i}{j}\psi_j = \ats{i}{j + 1}\psi_{j + 1}$. We sometimes refer to such parts of traces as ``silent steps''.
\item $\ats{i}{j}\psi_j = \ats{i}{j}(\theta_1 \vee \theta_2)$ is the principal formula in an application of the $\vee$-rule, and $\ats{i}{j + 1}\psi_{j + 1} \in \{\ats{i}{j}\theta_1, \ats{i}{j} \theta_2)\}$.
\item $\ats{i}{j}\psi_j = \ats{i}{j}(\theta_1 \wedge \theta_2)$ is the principal formula in an application of the $\wedge$-rule, and $\ats{i}{j + 1}\psi_{j + 1} = \{\ats{i}{j}\theta_1$ or  $\ats{i}{j + 1}\psi_{j + 1} = \{\ats{i}{j}\theta_2$ depending on whether $u_{j + 1}$ is the left or right premise of the rule.
\item $\ats{i}{j}\psi_j = \ats{i}{j}\at{i'}\theta$ is the principal formula in an application of the $\mathsf{Glob}$-rule, and $\ats{i}{j + 1}\psi_{j + 1} = \at{i'}\theta$.
\item $\ats{i}{j}\psi_j$ is the principal formula in an application of the $\mathsf{Eq}$-rule with side formula $\bfi_j \napprox \bfi'$, and $\ats{i}{j + 1}\psi_{j + 1} = \at{i'}\psi_j$.
\item $\ats{i}{j}\psi_j = \ats{i}{j}\eta x. \theta(x)$ is the principal formula in an application of the $\eta$-rule, and $\ats{i}{j + 1}\psi_{j + 1} = \ats{i}{j}\theta(\eta x. \theta(x))$. In this case we say that an \emph{unfolding} of variable $x$ occurs on the trace $\vec{t}$ at the index $j$.
\item  $u_j$ is the conclusion of an application of the $\mathsf{Mod}$-rule labelled $\Gamma, \at{i}\Box\theta,\at{i}\Diamond \Psi$, the premise is labelled $\Gamma, \at{j}\theta, \at{j}\Psi,\at{i}\Box\theta,\at{i}\Diamond \Psi$, $\ats{i}{j}\psi_j = \at{i}\Box \theta$, and $\ats{i}{j+1}\psi_{j+1} = \at{j} \theta$.
\item $u_j$ is the conclusion of an application of the $\mathsf{Mod}$-rule labelled $\Gamma, \at{i}\Box\theta,\at{i}\Diamond \Psi$, the premise is labelled $\Gamma, \at{j}\theta, \at{j}\Psi,\at{i}\Box\theta,\at{i}\Diamond \Psi$, and for some $\psi \in \Psi$,  $\ats{i}{j}\psi_j = \at{i}\Diamond \psi$ and $\ats{i}{j+1}\psi_{j+1} = \at{j} \psi$.
\end{enumerate}

We say that the infinite trace $\vec{t}$ is \emph{trivial} if for some $j \leq \omega$,  $\ats{i}{j}\psi_j = \ats{i}{m}\psi_{m}$ for all $j \leq m < \omega$. A non-trivial infinite trace is said to be \emph{good} if the highest ranking fixpoint variable that is unfolded infinitely many times on $\vec{t}$ is a $\nu$-variable. 
\end{defi}
Note that we do not require good traces to start at the root, but adding this constraint would make no substantial difference since every formula appearing in a sequent somewhere in a proof can be connected to a trace starting at the root.

\begin{defi}
An \sysinf-proof is said to be \emph{valid} if every infinite branch contains a good trace, and every leaf is labelled by an axiom.
\end{defi}

As mentioned, we need to be careful about how and when to apply the weakening rule to maintain an upper bound on the size of sequents. The following terminology will play an important role in this regard.

\begin{defi}
Given an $\sysinf$-proof $\Pi$ for some formula $\at{r}\rho$ (the ``root formula'' of the proof), a nominal $\bfj$ appearing in $\Pi$ is said to be  \emph{original} if it appears in $\at{r}\rho$. A formula appearing in $\Pi$ is said to be a \emph{ground formula} if it is of the form $\at{j}\psi$ where $\psi$ is an original nominal. 
\end{defi}

\begin{defi}
An \sysinf-proof is said to be \emph{frugal} if at most finitely many sequents appear in the proof. 
\end{defi}

\subsection{Derived rules}

We shall allow the use of derived rules in  proof constructions, as abbreviations of their derivations. In particular, for the $\mathsf{Mod}$-rule, we define what we will call its \emph{narrow} counterpart which is a derived rule of $\sysinf$. If the principal formula $\at{i}\Box \varphi$ is a ground formula then the rule is the same as $\mathsf{Mod}$. Otherwise, an instance of the narrow rule consists of an application of the modal rule  immediately followed by an application of the weakening rule in order to remove \emph{all} formulas of the form $\at{k}\psi$ that appear in the premise, and for which $\bfk$ is not an original nominal. For example, if $\bfi$ is a non-original nominal and $\bfj$ is original, then the following is an instance of the narrow $\mathsf{Mod}$-rule:
\begin{prooftree}
\AxiomC{$\at{k}\varphi, \at{k} \psi, \at{j}\theta$}
\UnaryInfC{$\at{i}\Box \varphi, \at{i}\Diamond \psi, \at{i} p, \at{j}\theta$}
\end{prooftree} 
If $\bfj$ is non-original then the corresponding instance would be:
\begin{prooftree}
\AxiomC{$\at{k}\varphi, \at{k} \psi$}
\UnaryInfC{$\at{i}\Box \varphi, \at{i}\Diamond \psi, \at{i} p, \at{j}\theta$}
\end{prooftree}

Note that what counts as an instance of the  narrow modal rule depends on what the root formula of a proof is, so these rules are not local in that sense. We therefore stress that the narrow modal rule is not explicitly part of the proof system, but just a derived rule that we will be using to simplify our reasoning.

\begin{conv}
Throughout the rest of the paper we fix an arbitrary strict linear order $\prec$ over all formulas (which restricts to an order over the set of nominals since each nominal is a formula), and for each given formula we fix an arbitrary strict linear order (also denoted by $\prec$) over the set of instances of rules in \sysinf{} in which that formula is the principal one.  This order will only be used as a book-keeping device to facilitate proofs, and has no substantial content. 
\end{conv}

In the next section we shall describe a game for constructing \sysinf-proofs between two players \plrval{} and \plrsat{}, where \plrval{} attempts to construct a valid proof and \plrsat{} attempts to show that no such proof can be constructed. This game will be formulated in terms of certain derived rules of \sysinf. A key idea will be to identify the crucial sources of non-determinism in the proof construction, in the sense that the only important choice available to \plrval{} is how to apply the modal rule, i.e. choosing which boxed formula to use in order to introduce a fresh nominal. Besides that, the proof construction is essentially deterministic, with one exception: we shall allow \plrval{} also to apply repeating rule applications in a non-deterministic manner, to introduce traces when needed. 

This informal idea will be captured by two derived rules, that we call the deterministic rule and the ground rule.  Like the narrow modal  rule, the exact shape of these rules will depend on extra parameters besides the formulas that appear in premise and conclusion.

\paragraph{The deterministic rule}

The \emph{deterministic rule}  is defined as follows: given a sequent $\Gamma$, if there are no applicable instances of the  $\wedge$-rule, the $\vee$-rule, $\mathsf{Glob}$ or the $\eta$-rules \emph{except repeating ones}, then the deterministic rule does not apply. Otherwise,  the deterministic rule applies uniquely as follows: we pick the $\prec$-smallest formula in $\Gamma$ which is the principal formula in an applicable non-repeating instance of one of these rules, we pick the smallest such rule instance for which it is the principal formula, and we apply that rule.

Note that if we repeatedly apply the deterministic rule starting from some sequent $\Gamma$ until it no longer applies, then this process must eventually terminate. The assumption that all formulas are guarded plays an important role here, without guardedness the process could go on indefinitely via fixpoint unfoldings.

\paragraph{The ground rule}

The \emph{ground rule} is designed to deterministically apply the $\mathsf{Eq}$-rule and the $\mathsf{Com}$-rule in the same way as the deterministic rule, but also to ensure that original nominals are given special treatment. It is defined as follows: we consider the original nominals appearing in a sequent $\Gamma$. If possible, apply the $\prec$-smallest applicable rule instance  for which one of the following conditions holds:
\begin{enumerate}
\item it is a non-repeating instance of the rule $\mathsf{Com}$ with principal formula $\bfi \napprox \bfj$, where both $\bfi$ and $\bfj$ are original nominals, or:
\item it is a non-repeating instance of the rule $\mathsf{Eq}$ with principal formula $\at{j}\varphi$ and side formula $\bfj \napprox \bfi$, where $\bfi$ is a $\prec$-minimal original nominal for which such a rule instance applies. 
\end{enumerate}
If there are no such rule instances available then the ground rule does not apply. Like the deterministic rule, the process of repeatedly applying the ground rule must eventually terminate.

\section{Completeness for \sysinf}

\subsection{A game for building \sysinf-proofs}

To prove completeness we shall make use of a game for constructing \sysinf-proofs, played between two players \plrval{} (the proponent) and \plrsat{} (the opponent). We fix a root formula $\at{r}\rho$, so that what counts as a narrow instance of the modal rule is defined relative to this root formula as before, and similarly with the deterministic rule and the ground rule.

\begin{defi}
An instance of the  weakening rule is called \emph{terminal} if its premise is an axiom.
\end{defi}

\begin{defi}
The \emph{\sysinf-game} is a board game, defined as usual by specifying its positions, player assignments and admissible moves for positions and winning conditions on full (finite or infinite) plays. 

\textbf{Positions:} Game positions are of two types: sequents, which belong to $\plrval$, and pairs of sequents, which belong to $\plrsat$. 

\textbf{Moves for \plrsat:} Given a position belonging to \plrsat{}, consisting of a pair of sequents, the player simply  chooses one of the sequents from the pair. 

\textbf{Moves for \plrval:} Given a position belonging to \plrval{}, consisting of a sequent $\Gamma$, if $\Gamma$ is an axiom then the game ends and $\plrval$ is declared the winner. Otherwise available moves are defined as follows: 
\begin{itemize}
\item If the deterministic rule is applicable to $\Gamma$ then this is the only move allowed for \plrval. 
\item If the deterministic rule is not applicable,  but the ground rule is applicable, then this is the only move allowed for \plrval. 
\item If neither the deterministic rule nor the ground rule are applicable, then the possible moves of \plrval{} are the narrow modal rule, terminal applications of the weakening rule,  and \emph{repeating} applications of non-modal rules.  
\end{itemize}
If $\pi$ is a partial play ending with some sequent $\Gamma$, then we often refer to $\Gamma$ as the \emph{label} of $\pi$.
Note that since we allow repeating applications of Weakening, \plrval{} never gets stuck. So the only full finite plays are those that end in an axiom, and are won by \plrval. Thus to finish the construction of the \sysinf-game it remains only to decide the winner of an infinite play. 
\emph{Traces} on a play of the \sysinf-game are defined similarly as traces in proof trees, the only difference being that a trace on a play $\pi$ of length $k \leq \omega$ is now an object of the form $(\pi_n,\ats{i}{n}\varphi_n)_{n < k}$ where each $\pi_n$ is an initial segment of the play $\pi$, and for each $n + 1 < k$ the initial segment $\pi_{n +1}$ extends $\pi_n$ with a single move. Given an infinite play $\pi$, \plrval{} is then declared the winner if the play contains a good trace, and otherwise the winner is \plrsat.  
\end{defi}

We now draw some simple consequences of how the \sysinf-game has been designed.

\begin{prop}
\label{p:only-original}
In any sequent of the form $\Gamma,\at{i}\psi$ appearing in a play of the \sysinf-game, $\psi$ contains no non-original nominals. 
\end{prop}

\begin{proof}
Just observe that all the admissible moves preserve this condition. 
\end{proof}

From this proposition a few useful facts follow:

 \begin{prop}
\label{p:identity}
If a play of the \sysinf-game contains any sequent of the form $\Gamma, \bfi \napprox \bfj$, then  $\bfj$ is an original nominal. 
\end{prop}
\begin{proof}
Special case of Proposition \ref{p:only-original}.
\end{proof}
\begin{prop}
\label{p:boundedpp}
For each nominal $\bfi$, and each partial  play $\pi$ in the $\sysinf$-game, the label of $\pi$ contains at most $k$ formulas of the form $\at{i}\psi$, where $k$ is linear in the size of the root formula.  
\end{prop}
\begin{proof}
Easy using Proposition \ref{p:only-original}. . 
\end{proof}

\begin{prop}
\label{p:narrow-game}
For any sequent $\Gamma$ appearing in a play of the $\sysinf$-game, at most one non-original nominal appears in $\Gamma$.
\end{prop}

\begin{proof}
The only moves that can introduce new non-original nominals are applications of the narrow modal rule, and by design each instance of this rule erases all occurrences of non-original nominals other than the new nominal that was introduced.  
\end{proof}

Like the games for satisfiability checking for the modal $\mu$-calculus introduced in \cite{niwi:game96}, the proof search game  is determinate:

\begin{prop}
\label{p:determinacy}
The \sysinf-game is determinate, i.e. for every position there is a player who has a winning strategy. 
\end{prop}

\begin{proof}
We want to prove this using Martin's Theorem \cite{mart:bore75}. In order for this result to apply, we want to represent the proof search game as a Gale-Stewart game with winning condition given by a Borel set. In \cite{niwi:game96}, this is proved by observing that the winning infinite plays form an $\omega$-regular language. The only obstacle to this approach here is that infinitely many sequents may appear in a play,  so we cannot use the set of sequents as the alphabet of a non-deterministic parity automaton recognizing winning plays. 

However, this issue is easily dealt with: by Proposition \ref{p:narrow-game}, the set of sequents that may appear in a play of the proof search game is finite up to renaming of nominals. Using this observation it is not hard to see that we can represent the proof search game as an equivalent Gale-Stewart game with finitely many positions, and in which the winning condition is an $\omega$-regular language. We leave the details of this construction to the reader.   
\end{proof}

%

\begin{defi}
Let $\sigma$ be a strategy for \plrsat{} in the \sysinf-game, and let $\pi$ be a partial play ending with a winning position. The \emph{$\sigma$-tail} of $\pi$, denoted $\sitail{\pi}$, is defined so that $\pi \cdot \sitail{\pi}$ is the unique longest $\sigma$-guided partial play extending $\pi$ and such that the only moves made by \plrval{} on $\sitail{\pi}$ are instances of the deterministic rule or the ground rule. 
\end{defi}
Note that the previous definition is sound, since the process of repeatedly applying admissible moves corresponding to the deterministic rule or the ground rule is entirely deterministic, and must eventually terminate. 

Note that, since the set of \emph{finite} partial plays in the \sysinf-game is a countable set (being a set of finite sequences over a countable set), we can define a surjective mapping $F$ from the set of nominals to the set of finite $\sigma$-guided partial plays, such that $F^{-1}[\pi]$ is infinite for each finite partial play $\pi$. We leave the little set theoretic exercise of proving this to the reader. Throughout the rest of this section we fix such a map $F$. Informally, we think of $F(\bfi)$ as a \emph{tag} attached to the nominal $\bfi$ to remember where it was introduced. 
\begin{defi}
We say that a full or partial play $\pi$ of the \sysinf-game is \emph{clean} if, for every initial segment $\pi'$ of the play ending with an application of the (narrow) modal rule introducing a new nominal $\bfj$, we have $F(\bfj) = \pi'$.
\end{defi}
When proving completeness of \sysinf{} we shall construct a counter-model to the root formula from a winning strategy for \plrsat, and it will then be convenient to restrict attention to clean plays. In particular, the definition of a clean play together with the design of the \sysinf-game yields the following result:

\begin{prop}
\label{p:linearity}
Let $\sigma$ be a given strategy for \plrsat{} in the \sysinf-game, and let $\bfi$ be a non-original nominal that appears on some clean $\sigma$-guided partial play $\pi$. Then there exists a unique \emph{largest} set of formulas of the form $\{\at{i}\varphi_1,...,\at{i}\varphi_n\}$, which we shall denote by $\mathsf{Th}(\bfi,\sigma)$, such that the label of some clean $\sigma$-guided partial play $\pi'$  contains $\mathsf{Th}(\bfi,\sigma)$. Furthermore, the label of the play $\pi \cdot \sitail{\pi}$ contains $\mathsf{Th}(\bfi,\sigma)$.
\end{prop}

\subsection{Trace loops}

A crucial part of proving completeness of \sysinf{} consists in proving that the standard ``good trace'' condition on valid proofs, in terms of traces going from the root up along a single branch, is not too strong. At first sight it may seem that we need to consider a more general condition, allowing traces to jump between different occurrences of the same nominal. In this subsection we prove a useful result that deals with this issue. 

\begin{defi}
Let $S$ be a set of plays in the \sysinf-game. A \emph{good trace loop} on $S$ in the \sysinf-game is a sequence of partial traces $\langle \vec{t}_1....\vec{t}_n \rangle$ for which there exist $\pi_1,...,\pi_n \in S$ such that:
\begin{itemize}
\item Each $\vec{t}_i$ is a partial trace on $\pi_i$,  
\item Each $\vec{t}_i$ starts and ends with  ground formulas,
\item Each  $\vec{t}_{i + 1}$ starts with the last formula of $\vec{t}_i$,
\item  The trace $\vec{t}_{1}$ starts with the last formula of $\vec{t}_n$, and
\item  At least one variable is unfolded on some trace $\vec{t}_i$ and the highest ranking such variable is a $\nu$-variable.
\end{itemize} 
\end{defi}

\begin{lemma}
\label{l:better-strategy}
Suppose that \plrsat{} has a winning strategy  in the \sysinf-game for $\at{i}\rho$. Then there exists a sequent $\Phi$ containing $\at{i}\rho$ and a winning strategy $\sigma$ for \plrsat{} in the \sysinf-game with starting position $\Phi$, such that the following conditions hold:
\begin{enumerate}
\item For every sequent appearing in a $\sigma$-guided play, its ground formulas are the same as the ground formulas in $\Phi$.
\item The set of $\sigma$-guided plays does not contain any good trace loops.
\end{enumerate}
 \end{lemma}

\begin{proof}
We first prove the following claim:
\begin{claim}
There exists some sequent $\Phi$ such that:
\begin{itemize}
\item $\at{r}\rho \in \Phi$,
\item \plrsat{} has a winning strategy $\sigma$ in the \sysinf-game at the starting position $\Phi$,
\item For every sequent $\Gamma$ that appears in some $\sigma$-guided partial play, the ground formulas appearing in $\Gamma$ are the same as the ground formulas in  $\Phi$.  
\end{itemize}
\end{claim}

\begin{pfclaim}
Let $\tau$ be the winning strategy assumed to exist for \plrsat. First note that the ground formulas appearing in $\tau$-guided partial plays in the \sysinf-game are increasing in the sense that, whenever $\Gamma'$ appears later than $\Gamma$ in a partial play, all ground formulas in $\Gamma$ are also in $\Gamma'$. This is  because the only admissible rule application that can remove a ground formula is a terminal application of the weakening rule, the premise of which is an axiom. Such applications of weakening never happen in any $\tau$-guided partial play, since such a play would be a loss for \plrsat.

We construct a series of $\tau$-guided partial plays $\pi_0,\pi_1,\pi_2...$, where each $\pi_i$ is an initial segment of $\pi_{i+1}$. For each $i$ we let $G_i$ be the set of ground formulas appearing on the last position of $\pi_i$. We shall maintain the invariant that, for all proper initial segments $\pi'$ of $\pi_{i + 1}$, the ground formulas appearing in the last sequent of $\pi'$ are contained in $G_{i}$.   Let $\pi_0$ be the start position of the \sysinf-game. Suppose that $\pi_i$ has been constructed. If there is no $\tau$-guided partial play $\pi'$ extending $\pi_i$ in which the last sequent contains ground formulas not in $G_i$, then we are done: for all $\tau$-guided partial plays extending this play, the ground formulas appearing in all sequents must be equal to $G_i$, and $\tau$ provides a winning strategy for $\plrsat$ in the \sysinf-game for the label of $\pi_i$. If there is some $\tau$-guided partial play $\pi'$ extending $\pi_i$ in which the last sequent contains ground formulas not in $G_i$, then just pick $\pi_{i+1}$ to be its smallest initial segment for which this holds. This procedure must eventually terminate, since otherwise we get an infinite and strictly increasing sequence of sets of ground formulas $G_0 \subset G_1 \subset G_2 ...$, which is impossible since there are only finitely many possible ground formulas by Proposition \ref{p:boundedpp}. 
\end{pfclaim}

Now let $\Phi$ and $\sigma$ be as in the previous claim. Given a $\sigma$-guided play $\pi$, let $\uparrow\pi$ be the set of partial plays $\pi'$ such that $\pi \cdot \pi'$ is a $\sigma$-guided partial play. Our aim is to find a $\sigma$-guided play $\pi$ such that $\uparrow \pi$ does not contain any good trace loops; we can then simply take the label of $\pi$ to the sequent claimed to exist in the statement of the Proposition, and we obtain the required winning strategy for \plrsat{} by assigning the move $\sigma(\pi \cdot  \pi')$ to a partial play $\pi'$.  

Given a good trace loop $\langle \vec{t}_1,...,\vec{t}_n\rangle$, let its \emph{kind} be the set of triples: 
$$\{(\ats{i}{1}\varphi_1, X_1, \ats{j}{1}\psi_1),...,(\ats{i}{n}\varphi_n, X_n, \ats{j}{n}\psi_n)\}$$
such that for each $m \in \{1,...,n\}$, the trace $\vec{t}_m$ begins with $\ats{i}{m}\varphi_m$, ends with $\ats{j}{m}\psi_m$, and the variables unfolded on $\vec{t}_m$ are precisely the members of the set $X_m$. Since each trace in a good trace loop begins and ends with a ground formula, and since there are only finitely many ground formulas, there are finitely many kinds of good trace loops. We shall show how to find a $\sigma$-guided play $\pi$ such that $\uparrow \pi$ does not contain any good trace loops of a given kind. By simply repeating the argument, we can then kill off all the kinds of good trace loop one by one.

So let the kind $K$ be $\{(\ats{i}{0}\varphi_0, X_0, \ats{j}{0}\psi_0),...,(\ats{i}{n-1}\varphi_{n-1}, X_{n-1}, \ats{j}{n-1}\psi_{n-1})\}$. We construct a sequence of partial plays $\pi_0 \sqsubseteq \pi_1 \sqsubseteq \pi_2$... as follows. If the set of all $\sigma$-guided plays does not contain any good trace loops of kind $K$, we are done.  Otherwise, let $\pi_0$ be some play on which the part $(\ats{i}{0}\varphi_0, X_0, \ats{j}{0}\psi_0)$ appears, which must exist. Note that we have a partial trace $\vec{t}_0$ on $\pi_0$ leading from  $\ats{i}{0}\varphi_0$ to $ \ats{j}{0}\psi_0$ on which exactly the variables $X_0$ were unfolded; since the first formula is a ground formula, and these are the same in all positions in all $\sigma$-guided plays, we can simply ``pad'' the partial trace with silent steps repeating the same formula to extend it to a trace on the whole play $\pi_0$. Now we repeat the procedure: if $\uparrow\pi_0$ does not contain any good trace loop, then we are done. Otherwise, we can extend $\pi_0$ in the same way to a partial play $\pi_1$ containing a trace $\vec{t}_1$ appearing \emph{after} $\pi_0$, such that $\vec{t}_1$ starts with $\ats{i}{1}\varphi_1$, ends with $\ats{j}{1}\psi_1$ and the variables unfolded are precisely $X_1$. Then since $\vec{t}_0 $ and $ \vec{t}_1$ end and start respectively with the same ground formula, and ground formulas stay the same, they can be linked together by ``padding with silent steps'' repeating this formula to form a trace on $\pi_1$. The idea is visualized in Figure \ref{f:loops}, showing instances of trace loops of the same kind, containing two types of traces represented by solid and dashed lines respectively. The shaded lines represent ``silent traces'' repeating the formulas $\mathsf{A}$ and $\mathsf{B}$ respectively to connect the fragments to a trace, and the shaded area within each sequent shows the part of the sequent consisting of ground formulas. 
\begin{figure}
\begin{center}
\includegraphics[width=8cm, height=8cm]{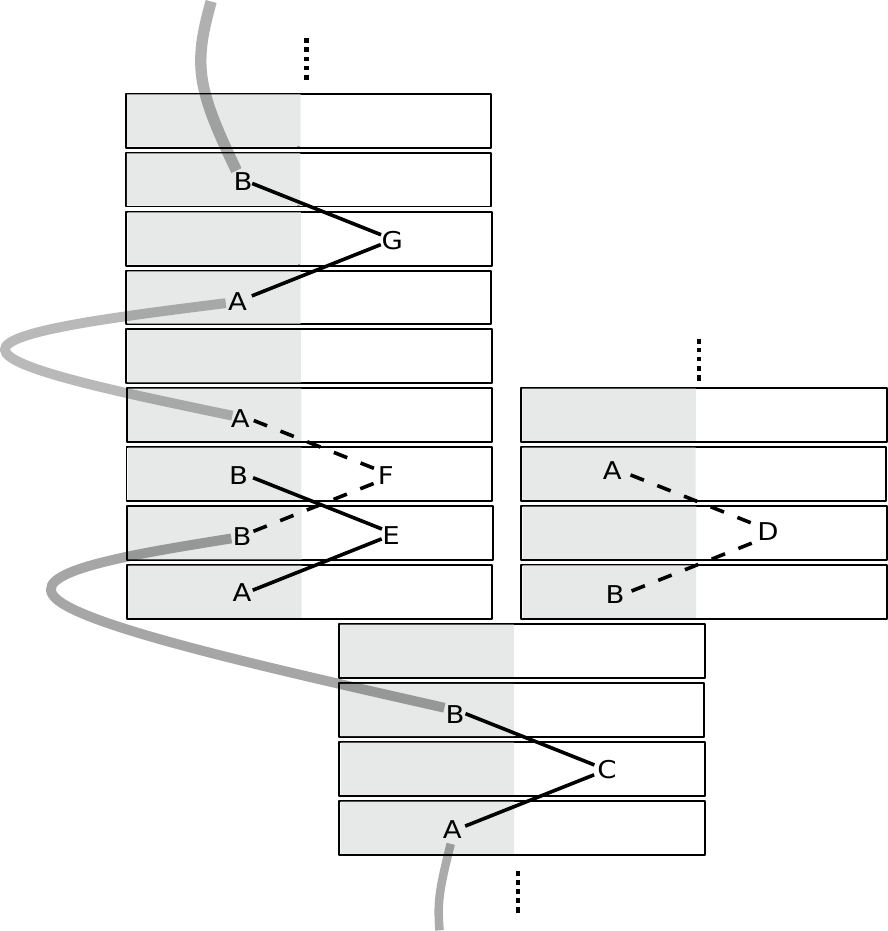}
\end{center}
\caption{Scattered fragments of trace loops are connected to form a trace.}
\label{f:loops}
\end{figure}
It is not hard to see that, if this procedure never terminates, then we end up building an infinite $\sigma$-guided play containing a good trace, which is a contradiction since $\sigma$ was a winning strategy. So the procedure terminates with some $\pi_m$, and the proof is finished.  
\end{proof}

\subsection{Completeness}

We are now ready to prove soundness and completeness of \sysinf.

\begin{theo}
\label{t:infcompl}
Let $\rho$ be any formula. The following are equivalent: (a) $\rho$ is valid, (b) $\plrval$ has a winning strategy in the \sysinf-game for $\at{r} \rho$, where $\mathsf{r}$ is some fresh nominal, (c) $\rho$ has a valid and frugal \sysinf-proof, (d) $\rho$ has a valid \sysinf-proof.
\end{theo}

\begin{proof}
We prove (a) $\Rightarrow$ (b) $\Rightarrow$ (c) $\Rightarrow$ (d) $\Rightarrow$ (a). 

\medskip

\textbf{(b) $\Rightarrow$ (c) $\Rightarrow$ (d):} The step (c) $\Rightarrow$ (d) is trivial. For $(b) \Rightarrow (c)$, it is easy to see that we can read off an \sysinf-proof for the root formula from any winning strategy for \plrval{} in the \sysinf-game. By Proposition \ref{p:narrow-game}, every sequent in this proof will contain at most one non-original nominal. So by suitably renaming nominals we can easily turn the proof into a frugal one. 

\medskip

\textbf{(a) $\Rightarrow$ (b):} We prove this by contraposition. Suppose  there is a winning strategy for \plrsat{} in the \sysinf-game for $\at{r}\rho$. Let $\Phi$ be a set of ground formulas containing $\at{r}\rho$ and let $\sigma$ be a winning strategy for $\plrsat$ in the \sysinf-game for premise $\Phi$ such that the ground formulas stay the same in every $\sigma$-guided play, and the set of $\sigma$-guided plays contains no good trace loops. We shall construct a countermodel to (the disjunction of) $\Phi$, which gives a countermodel to $\rho$ since $\at{r}\rho \in \Phi$.

We construct the model $M = (W,R,A,V)$  using the strategy $\sigma$ as follows.  Let $N$ be the set of nominals $\bfi$ such that $\bfi$ appears in some position in some clean $\sigma$-guided play $\pi$, and let $\equiv$ be the smallest equivalence relation over $N$ containing all pairs $(\bfi,\bfj)$ for which $\bfi \napprox \bfj$  appears in some position in some  clean $\sigma$-guided play $\pi$. 
\begin{claim}
For each $\bfi$, the equivalence class $[\bfi]$ modulo $\equiv$ is either a singleton or contains at least one of the nominals in $\varphi$.
\end{claim}

\begin{pfclaim}
By Proposition \ref{p:identity}.
\end{pfclaim}

Motivated by this claim, we call a nominal \emph{representative} if its equivalence class is a singleton, or it is the $\prec$-smallest original nominal belonging to its equivalence class. We let $W$ be the set of representative members of $N$. We set $\bfi R \bfj$ iff there is some $\bfj' \equiv \bfj$ and a clean $\sigma$-guided play in which $\bfj'$ is introduced by an application of the modal rule to the nominal $\bfi$. Set $A(\bfi)$ to be the representative of $[\bfi]$. Finally, for a representative $\bfi$ set $\bfi \in V(p)$ iff $\at{i} \neg p$ appears on some clean $\sigma$-guided play. We shall show that $M$ is a counter-model to the sequent $\Phi$. First we prove a few auxiliary claims. 

Let $\sim$ be the binary relation over $W$ defined by setting $\bfi \sim \bfj$ iff the formula $\bfi \napprox \bfj$ appears on some clean $\sigma$-guided play (or equivalently, on all positions in all $\sigma$-guided plays).
\begin{claim}
\label{simgood}
The restriction of the relation $\sim$ to the original nominals in $W$ is symmetric and transitive.
\end{claim}

\begin{pfclaim}
Symmetry follows immediately since some instance of the ground rule will eventually apply in the form of  $\mathsf{Com}$, given that $\bfi \sim \bfj$ are original nominals. Transitivity follows in the same way, but this time using the ground rule in the form of $\mathsf{Eq}$. 
\end{pfclaim}

\begin{claim}
\label{eqsim}
Given two original nominals $\bfi, \bfj$, if $\bfi \equiv \bfj$ then $\bfi \sim \bfj$. 
\end{claim}

\begin{pfclaim}
Since $\bfi \equiv \bfj$, there is a tuple $(\bfk_1,...,\bfk_n)$ of nominals such that $\bfk_1 = \bfi$, $\bfk_n = \bfj$, and for each $m < n$ either $\bfk_{m} \napprox \bfk_{m + 1}$ or $\bfk_{m + 1} \napprox \bfk_m$ appears on some clean $\sigma$-guided play.  By Proposition \ref{p:identity}, no two consecutive nominals in the list can both be non-original nominals. So by  the transitivity part of Claim \ref{simgood}, it suffices to prove the following two items:
\begin{enumerate}
\item If $m < n$, and $\bfk_m,\bfk_{m+1}$ are both original nominals, then $\bfk_m \sim \bfk_{m+1}$.
\item If $m + 1 < n$, $\bfk_m$ and $\bfk_{m +2}$ are original nominals and $\bfk_{m + 1}$ a non-original nominal, then $\bfk_m \sim \bfk_{m + 2}$. 
\end{enumerate}
Item (1) is immediate from the symmetry part of Claim \ref{simgood}. For item (2), let $\bfk_m,\bfk_{m + 1}, \bfk_{m + 2}$ be as described. Since $\bfk_{m + 1}$ is a non-original nominal, by Proposition \ref{p:identity} the only possibility is that the formulas $\bfk_{m + 1} \napprox \bfk_m$ and $\bfk_{m + 1} \napprox \bfk_{m + 2}$ both appear on clean $\sigma$-guided plays. By definition of a $\sigma$-guided play, and since no nominal is introduced twice in two clean $\sigma$-guided plays, this is only possible if they appear on the same position in some $\sigma$-guided play  $\pi$. Say that $\bfk_{m} \prec \bfk_{m + 2}$, since the other case is symmetric. But then \plrval{}  will eventually have to play an application of the ground rule leading to a node whose label contains $\bfk_m \napprox \bfk_{m + 2}$. In fact since the ground formulas are the same in all labels, this means $\bfk_m \napprox \bfk_{m + 2}$ must belong to the label of every play. So $\bfk_m \sim \bfk_{m + 2}$ as required. 
\end{pfclaim}

\begin{claim}
\label{transfer}
Suppose that $\bfi \equiv \bfj$ and that $\bfi,\bfj$ are original nominals. Then for any position $u$ appearing in a clean $\sigma$-guided play, and any $\theta$, if $\at{i}\theta$ belongs to $u$ then so does $\at{j}\theta$. 
\end{claim}

\begin{pfclaim}
By Claim \ref{eqsim} we have $\bfi \sim \bfj$, hence  $u$ is of the form $\Gamma,\at{i}\theta,\bfi \napprox \bfj$. If $\at{j}\theta$ is not in the label of $u$ then \plrval{} must eventually play the ground rule leading to a node whose label contains $\at{j}\theta$, and this contradicts the constancy of the ground formulas in all positions appearing on $\sigma$-guided plays. 
\end{pfclaim}

\begin{claim}
\label{strong-transfer}
Suppose that $\bfi \equiv \bfj$ and that $\bfj$ is an original nominal. Then for any basic position $u$ appearing in a clean $\sigma$-guided play, and any $\theta$, if $\at{i}\theta$ belongs to $u$ then so does $\at{j}\theta$. 
\end{claim}

\begin{pfclaim}
By Claim \ref{transfer} we only need to consider the case where $\bfi$ is non-original. Suppose $\pi$ is a clean $\sigma$-guided partial play whose label contains $\at{i}\theta$ and suppose $\bfi \equiv \bfj$. Clearly there must be some nominal $\bfj' \sim \bfj$ such that $\bfi \napprox \bfj'$ appears in some clean $\sigma$-guided play, and by Proposition \ref{p:only-original} the nominal $\bfj'$ must be original.  By Proposition \ref{p:linearity} $\bfi \napprox \bfj$ and $\at{i}\theta$ both belong to the label of $\sitail{\pi}$, and it follows using the ground rule and constancy of ground formulas that $\at{j'}\theta$ must belong to the label of $\sitail{\pi} $ and hence of $\pi$. Since $\bfj' \sim \bfj$ it follows by Claim \ref{transfer} that $\at{j}\theta$ belongs to the label of $\pi$. 
\end{pfclaim}

\begin{claim}
\label{finishtrace}
Suppose that $\vec{t}$ is some partial trace on a clean partial $\sigma$-guided play $\pi$, such that the last element of the trace $\vec{t}$ is of the form $(\pi,\at{k'}\psi)$ where $A(\bfk') = \bfk$.  If $\psi$ is of the form $\Box \theta$ or $\Diamond \theta$, then there is a clean $\sigma$-guided play $\upsilon$ extending $\pi$ and a partial trace on $\upsilon$ of the form $(\pi,\at{k'}\psi) \cdot \vec{u} \cdot (\upsilon,\at{k}\psi)$, which contains no fixpoint unfoldings. 
\end{claim}

\begin{pfclaim}
The interesting case is when $\bfk' \neq \bfk$, in which case $\bfk$ must be an original nominal.  
Since box- and diamond-formulas are never decomposed by the derministic rule or the ground rule, it is easy to see (using Proposition \ref{p:linearity}) that $\sitail{\pi}$ contains a trivial trace ending with the same formula $\at{k'}\psi$, and there must be some original nominal $\bfk'' \sim \bfk$ such that $\bfk' \equiv \bfk''$. By Proposition \ref{p:linearity} again, $\bfk' \napprox \bfk''$ belongs to the label of $\sitail{\pi}$, and so does $\bfk'' \napprox \bfk$ by constancy of the ground formulas.  By Claim \ref{strong-transfer} the formula $\at{k''}\psi$ belongs to the label of $\sitail{\pi}$, hence $\at{k}\psi$ does, and we can now let \plrval{} play a repeating application of the  $\mathsf{Eq}$-rule to construct a trace ending with this formula. 
\end{pfclaim}

We now proceed to show that the sequent $\Phi$ is not valid in $M$.  
Pick any formula $\at{i}\varphi \in \Phi$. We shall construct a winning strategy $\sigma'$ for \plrsat{} in the evaluation game  in $M$ at the starting position $(A(\bfi),\varphi)$.  Inductively, as an invariant we associate  with each partial $\sigma'$-guided partial play of $\pi'$ of the form:
$$(\bfj_1, \psi_1) \cdot \vec{p} \cdot (\bfj_n,\psi_n)$$
a sequence of non-empty partial traces $\langle\vec{t}_1,...,\vec{t}_n\rangle$ such that each of these traces $\vec{t}_k$ belongs to some clean $\sigma$-guided partial play $\pi_k$, and such that the following conditions hold:
\begin{description}
\item[I1:] The last element of each trace $\vec{t}_k$ is of the form $(\pi_k,\at{j_k'}\psi_k)$ where $A(\bfj_k') = \bfj_k$.  Furthermore, if $\psi_k$ is of the form $\Box \theta$ or $\Diamond \theta$, then $\bfj_k' = \bfj_k$. 
\item[I2:] For each $k < n$, if the last element of $\vec{t}_k$ is $(\pi_k,\at{j_k'}\psi_k)$ then the  first element of $\vec{t}_{k + 1}$ is of the form $(\pi_{k + 1},\at{j_k'}\psi_k)$. Furthermore, if $\bfj_k$ is not an original nominal then $\pi_k = \pi_{k + 1}$.
\item[I3:] For $k < n$, a fixpoint unfolding occurs on the trace $\vec{t}_{k + 1}$ iff the same fixpoint is unfolded on $(\bfj_k,\psi_k)\cdot (\bfj_{k+1},\psi_{k + 1})$.   
\end{description}
Suppose we are given a clean $\sigma'$-guided partial play $\pi'$ of the form $(\bfj_1, \psi_1) \cdot \vec{p} \cdot (\bfj_n,\psi_n)$, and that the associated sequence of partial traces $\langle\vec{t}_1,...,\vec{t}_n\rangle$ has been constructed. We shall show that if the last position on $\pi'$ belongs to \plrsat, then we can define a move for which the invariant (I1) -- (I3) can be maintained, and if the last position belongs to \plrval{} then the invariant can be maintained for any possible move. This is proved by a case distinction as to the shape of the last position.

\paragraph{Case: $\psi_n$ is a literal.}

In this case, there are no possible moves, so we need to check that the losing player is not \plrsat. This happens in four possible cases:
\begin{enumerate}
\item $\psi_n = p$ and $\bfj_n \in V(p)$.
\item $\psi_n = \neg p$ and $\bfj_n \notin V(p)$.
\item $\psi_n = \bfi$ and $\bfj_n = A(\bfi)$.
\item $\psi_n = \neg \bfi$ and $\bfj_n \neq A(\bfi)$.
\end{enumerate}

In case (1), since $\bfj_n \in V(p)$ it must hold that there is some $\sigma$-guided play  containing $\ats{j}{n}\neg p$. But then, since this is a ground formula,  $\ats{j}{n}\neg p$ must belong to every position in every $\sigma$-guided play. In particular, the position $u_n$ contains $\ats{j}{n}\neg p$. But the last element of the trace $\vec{t}_n$ is of the form $(\pi_n,\ats{j'}{n} p)$ where $A(\bfj_n') = \bfj_n$. Hence $\bfj_n \equiv \bfj'_n$, and by Claim \ref{strong-transfer} the label of $\pi_n$ also contains $\ats{j}{n} p$. Hence it contains the axiom $\ats{j}{n}p, \ats{j}{n}\neg p$. A terminal application of  weakening gives a lost $\sigma$-guided play, which is a contradiction. In case (2), we get that $\ats{j'}{n} \neg p $ belongs to the label of $\pi_n$, and since again we have $\bfj_n \equiv \bfj'_n$ it follows by Claim \ref{strong-transfer} that  the label of $\pi_n$ also contains $\ats{j}{n} \neg p$. So $\bfj_n \in V(p)$ by definition of $V$, contradiction.
In case (3), the last element of the trace $\vec{t}_n$ is of the form $(\pi_n,\ats{j'}{n} \bfi)$ where $A(\bfj_n') = \bfj_n = A(\bfi)$. So $\bfj'_n \equiv \bfi$. Since the label of $\pi_n$ contains $\ats{\bfj'}{n} \bfi$, and since $\bfi$ is an original nominal,   the label of $\pi_n$ contains $\at{i}\bfi$ also. But this is an axiom, so  a terminal application of  weakening gives a lost $\sigma$-guided play, contradiction.
In case (4), we get that $\ats{j'}{n} \neg \bfi $ belongs to the label of $\pi_n$, and the short-hand for this formula is $\bfj'_n \napprox \bfi$. So $\bfj'_n \equiv \bfi$. By assumption $A(\bfj_n') = \bfj_n$, so $\bfj_n \equiv \bfj'_n \equiv \bfi$. Since $\bfj_n$ is a representative nominal, this is only possible if $\bfj_n = A(\bfi)$, contradiction. 

\paragraph{Case: $\psi_n = \theta \vee \theta'$.}  The last position of  $\pi'$ then belongs to $\plrval$. There are two possible $\sigma'$-guided continuations of $\pi'$, one for each disjunct. We consider how to maintain the invariant (I1)-- (I3) for the case of the first disjunct, since the other case is treated the same way.  The trace $\vec{t}_n$ ends with $(\pi_n, \ats{j'}{n}(\theta \vee \theta'))$ where the label of $\pi_n$ is some sequent $\Gamma, \ats{j'}{n}(\theta \vee \theta')$. We may assume without loss of generality that the trace $\vec{t}_n$ has been chosen so that $\sitail{\pi_n}$ does not contain any longer trace of the form:
$$(\pi_n, \ats{j'}{n}(\theta \vee \theta')) \cdot (\pi_n \cdot u_1, \ats{k}{1}(\theta \vee \theta'))\cdot ... \cdot (\pi_n \cdot u_1 \cdot ...\cdot u_m,\ats{k}{m}(\theta \vee \theta'))$$
such that $A(\bfj'_n) = A(\bfk_1) = ... = A(\bfk_m)$, since if such a trace exists then a longest one exists in $\sitail{\pi_n}$, and this trace still satisfies the invariant (I1) -- (I3).

With this assumption in case, consider two cases: either $\sitail{\pi_n}$ is empty, or not. In the first case, the deterministic rule cannot apply  since it must be applied first whenever possible, and this means that the $\vee$-rule only applies as a repeating rule. So $\ats{j'}{n}\theta$ and $\ats{j'}{n}\theta'$ both belong to the label of $\pi_n$ already. The ground rule cannot apply either, since it would produce a longer trace of the previously described form in $\sitail{\pi_n}$, and we assumed $\vec{t}_n$ was the longest such trace. So we are free to let \plrval play a repeating instance of the $\vee$-rule, producing a trace $(\pi_n,\ats{j'}{n}(\theta \vee \theta')) \cdot (\pi_n \cdot u, \ats{j'}{n}\theta)$. In the case where $\sitail{\pi_n}$ is non-empty, by our assumption on the trace $\vec{t}_n$ the only possibility is that the $\vee$-rule is immediately applied to $\ats{j'}{n}(\theta \vee \theta')$, yielding again a trace  $(\pi_n,\ats{j'}{n}(\theta \vee \theta')) \cdot (\pi_n \cdot u, \ats{j'}{n}\theta)$.

In either case, we want to add the trace to the sequence $\langle \vec{t}_1,...,\vec{t}_n\rangle$ so that the invariant (I1) -- (I3) is maintained. The conditions (I2) and (I3) are obviously preserved. The only possible problem occurs if $\theta$ is a box- or diamond-formula, but $\bfj_n' \neq \bfj_n$. In this case we appeal to Claim \ref{finishtrace} to extend the trace to a longer partial trace that ends with the formula $\ats{j}{n}\theta$, and still  does not contain any fixpoint unfoldings.

\paragraph{Case: $\psi_n = \theta \wedge \theta'$.} The last position of  $\pi'$ belongs to $\plrval$. By assumption the trace $\vec{t}_n$ ends with $(\pi_n, \ats{j'}{n}(\theta \wedge \theta'))$. We make a similar assumption as before, that $\vec{t}_n$ has been chosen so that   $\sitail{\pi_n}$ does not contain any longer trace of the form:
$$(\pi_n, \ats{j'}{n}(\theta \wedge \theta')) \cdot (\pi_n \cdot u_1, \ats{k}{1}(\theta \wedge \theta'))\cdot ... \cdot (\pi_n \cdot u_1 \cdot ...\cdot u_m,\ats{k}{m}(\theta \wedge \theta'))$$
such that $A(\bfj'_n) = A(\bfk_1) = ... = A(\bfk_m)$. If $\sitail{\pi_n}$ is empty, then only repeating applications of the $\wedge$-rule can apply and the deterministic rule and ground rule do not apply. So we can let \plrval{} play a repeating application of the $\wedge$-rule. Otherwise, the $\wedge$-rule is applied immediately in $\sitail{\pi_n}$. In either case the strategy $\sigma$ determines a unique $\sigma$-guided play $\pi_n \cdot u$ containing a trace either of the form $(\pi_n, \ats{j'}{n}(\theta \wedge \theta'))\cdot(\pi_n \cdot u, \ats{j'}{n}\theta)$ or of the form $(\pi_n, \ats{j'}{n}(\theta \wedge \theta'))\cdot(\pi_n \cdot u, \ats{j'}{n}\theta')$, and containing no fixpoint unfoldings.  In the first case we extend $\sigma'$ by taking $\sigma'(\pi') = (\bfj_n, \theta)$, and in the second case we take $\sigma'(\pi') = (\bfj_n, \theta')$. Finally, we appeal to Claim \ref{finishtrace} as before to find a possibly longer trace that can be added to maintain the invariant (I1)-- (I3).

\paragraph{Case: $\psi_n = \at{j}\theta$.} There is only one possible move, and only one $\sigma'$-guided continuation to consider. The argument is therefore a simpler version of the previous cases.  

\paragraph{Case: $\psi_n = \eta z. \theta(z)$.} Again there is only one $\sigma'$-guided continuation to consider, and we maintain the invariant (I1)-- (I3)  using Claim \ref{finishtrace}. We just need to note in this case that the shadow trace will contain an unfolding of the variable $z$. 

\paragraph{Case: $\psi_n = \Box \theta$.} The last position of  $\pi'$ belongs to \plrsat. By (I1), the last element of $\vec{t}_n$ is of the form $(\pi_n,\ats{j'}{n}\Box \theta)$ where $A(\bfj'_n) = \bfj_n$, and by (I2) we have $\bfj'_n = \bfj_n$. We can assume that the deterministic and ground rules and conjunction rules do not apply to the label of $\pi_n$, since otherwise we can follow $\sitail{\pi_n}$ until this holds.  We then continue as follows: consider an arbitrary play by \plrval{} applying the narrow modal rule to $\ats{j}{n}\Box \theta$, leading to a sequent containing $\at{k}{\theta}$, where the nominal $\bfk$ is chosen to make sure that the play is clean.  Let this extended play be called $\upsilon$. It is clear that $\bfj_n R A(\bfk)$, so we extend the strategy $\sigma'$ by letting \plrsat{} play $(A(\bfk),\theta)$. We can now continue the trace leading to $(\upsilon,\at{k}\theta)$ to maintain the invariants (I1) -- (I3)  using  Claim \ref{finishtrace} in the same manner as before. 

\paragraph{Case: $\psi_n = \Diamond \theta$.} The last position of  $\pi'$ belongs to \plrval. By (I1), the last element of $\vec{t}_n$ is of the form $(\pi_n,\ats{j'}{n}\Diamond \theta)$ where $A(\bfj'_n) = \bfj_n$, and by (I2) $\bfj'_n = \bfj_n$. Suppose that \plrval{} plays $(\bfk,\theta)$ where $\bfj_n R \bfk$. Then there must be some $\bfk' \in [\bfk]$ and some play $\upsilon$ such that $\bfk'$ is introduced in $\upsilon$ by an application of the modal rule to the nominal $\bfj_n$, the premise of which we take to be the last element of $\upsilon$, and we let $\upsilon'$ be the part of $\upsilon$ leading to and including the conclusion.   We make a case distinction as to whether $\bfj_n$ is an original nominal. 

If $\bfj_n$ is an original nominal, then since $\ats{j}{n}\Diamond \theta$ is a ground formula, it belongs to the label of every play, including the label of $\upsilon'$. This means that $\upsilon$ must have a partial trace of the form $(\upsilon',\ats{j}{n}\Diamond \theta) \cdot (\upsilon,\at{k}'\theta)$. We continue the play $\upsilon$ using  Claim \ref{finishtrace} as before to find an extended trace $\vec{t}_{n+1}$ satisfying the invariant (I1) -- (I3).

If $\bfj_n$ is not an original nominal, then  since the narrow modal rule deletes all non-original nominals except the one it introduces, this means that $\bfk'$ must be introduced at the last step in some play $\upsilon$ that extends $\pi_n$, on which the nominal $\bfj_n$ must survive until the corresponding application of the modal rule. So $\ats{j}{n}\Diamond \theta$ must belong to the conclusion of the rule application. So we can find a trace of the form $(\pi_n,\ats{j}{n}\Diamond \theta) \cdot \vec{u} \cdot (\upsilon,\at{k'}\theta)$ on $\upsilon$ (containing no fixpoint unfoldings). We extend this trace to a suitable shadow trace $\vec{t}_{n +1}$ satisfying the invariant (I1) -- (I3) using  Claim \ref{finishtrace} as before.

\medskip
To finish the proof of (a) $\Rightarrow$ (b), we have given a strategy $\sigma'$ to \plrsat{} in the evaluation game such that the invariant (I1) -- (I3) is maintained. The strategy $\sigma'$ ensures that \plrsat{} never gets stuck, and any lost infinite $\sigma'$-guided play is easily seen to produce either an infinite clean $\sigma$-guided shadow-play in the \sysinf-game containing a good infinite trace, or a good trace loop on the set of all clean $\sigma$-guided plays. In either case we get a contradiction, so $\sigma'$ is winning for \plrsat{} and therefore we have found a falsifying model for $\rho$. 

%
%
\medskip

\textbf{(d) $\Rightarrow$ (a):}
Suppose there exists a valid \sysinf-proof $\Pi$ for root formula $\at{i}\rho$. We assume without loss of generality that every non-original nominal appearing in $\Pi$ has a unique application of the modal rule associated with it, i.e. no nominal is introduced by two different applications of the modal rule. Any proof can be put in this normal form by suitably renaming nominals.  Given a model $M = (W,R,A,V)$, suppose for a contradiction that there exists a strategy $\sigma$ for \plrsat{} in the evaluation game for $\rho$ in $M$ which is winning at starting position $(w_0,\rho)$. By Proposition \ref{p:parity-games} we can assume that the strategy $\sigma$ is positional. We may also assume  that $A$ is only defined for the nominals $N$ appearing in $\rho$, since the value of other nominals does not affect the truth value of $\rho$. 

For each $k$ with $1 \leq k < \omega$ we shall construct the following data: 
\begin{enumerate}
\item A node $u_k$ in the proof tree $\Pi$. We will maintain the invariant that $u_{k + 1}$ is always 
a child of the node $u_k$.
\item An assignment $A_{k}$ extending $A$ with values for all non-original nominals appearing on the branch up to $u_k$. We maintain the invariant that assignments associated with any two $n,m < \omega$ agree on all nominals for which both are defined, and that for each  formula in the label of $u_k$ of the form $\at{i}\psi$ such that $\psi$ is a literal, we have $$M_k, A_k(\bfi) \nVdash \psi$$ where $M_k = (W,R,A_k,V)$.
\item For each partial trace $\vec{t} = (u_1,\ats{i}{1} \psi_1) ... (u_k,\ats{i}{k} \psi_k)$ from the root of $\Pi$ to $u_k$ \emph{such that $\psi_k$ is not a literal}, a $\sigma$-guided partial play $\pi_{\vec{t}} = (v_1, \psi_1)...(v_k,\psi_k)$ (possibly with some ``silent steps'' repeating the same position) such that for each $j \leq k$, $A_j(\bfi_j) = v_j$. Note that  it follows that if $\bfi_j$ is an original nominal then $A_j(\bfi_j) = v_j$.  We will maintain the invariant that these partial plays ``commute with restrictions'' in the sense that the play associated with an initial segment of a trace will be an initial segment of the play associated with that trace, and that the same fixpoint unfoldings occur the same number of times on $\vec{t}$ and $\pi_{\vec{t}}$.
\end{enumerate}
To start the induction, we take $u_1$ to be the root node $r$ of $\Pi$ whose only trace is the singleton trace $\vec{t}$ consisting of  $(r,\at{r}\rho)$. We set $A_1(\mathsf{r}) = w_0$, and $\pi_{\vec{t}}$ is the singleton partial play consisting of the position $(w_0,\rho)$. Given that $u_k$ has been defined, we first note that $u_k$ cannot be a leaf, since it would then have to be an axiom, and in this case we easily get a contradiction:  if the leaf is labelled $\bfi \approx \bfi$ then by the inductive hypothesis we get $M_k,A_k(\bfi) \nVdash \bfi$ which is impossible. If it is labelled with $\at{i}p,\at{i}\neg p$ then by the induction hypothesis we get $M_k,A_k(\bfi) \nVdash p$ and $M_k,A_k(\bfi)\nVdash \neg p$, which is impossible. 

So $u_k$ must be the conclusion of a rule application, and we make a case distinction as to which rule is applied. In each case we only consider the extensions of traces resulting from the rule application, and let it be understood that shadow plays corresponding to traces that simply repeat the last element are extended by a ``silent step''. 

\paragraph{Unfolding rule:}
This case is easy: just extend the shadow play for the trace leading to the principal formula by a fixpoint unfolding, and take $A_{k+1} = A_k$.

\paragraph{Weakening:} Trivial. 

\paragraph{Rule $\mathsf{Eq}:$}
In this case the node $u_k$ is labelled with $\Gamma,\at{i}\psi,\bfi \napprox \bfj$ and we pick $u_{k + 1}$ to be the unique child labelled $\Gamma,\at{i}\psi,\at{j}\psi,\bfi \napprox \psi$. We set $A_{k + 1} = A_k$. If $\psi$ is not a literal we reason as follows: since $\neg \bfj$ is a literal we have by the induction hypothesis $M_k,A_k(\bfi) \nVdash \neg \bfj$ so  $A_{k}(\bfi) = A_k(\bfj)$, and the shadow play corresponding to a new trace ending with $(u_{k+1},\at{j}\psi)$ can be taken to extend the shadow play for the initial segment ending with $(u_k,\at{i}\psi)$ with a silent step. If $\psi$ is a literal, then we have $A_{k}(\bfi) = A_k(\bfj)$ again, and $M_k,A_k(\bfi) \nVdash \psi$ by the induction hypothesis, hence $M_k,A_k(\bfj)\nVdash \psi$ as well. 

\paragraph{Rule $\mathsf{Com}$:}
In this case $u_k$ is labelled $\Gamma,\bfi \napprox \bfj$ and the induction hypothesis gives $M_k,A_k(\bfi) \nVdash \neg \bfj$, so $A_k(\bfi) = A_k(\bfj)$. We take $u_{k + 1}$ to be the unique child labelled $\Gamma,\bfi \napprox \bfj, \bfj \napprox \bfi$, and the induction hypothesis clearly carries over with $A_{k + 1} = A_k$. 

\paragraph{Rule $\vee$:}
This rather trivial case is left to the reader. 

\paragraph{Rule $\wedge$:}
In this case the label of $u_k$ is $\Gamma,\at{i}(\varphi \wedge \psi)$. Since $\varphi \wedge \psi$ is not a literal, every trace leading to the principal formula of the rule application has a $\sigma$-guided shadow play. So the position $(A_k(\bfi),\varphi \wedge \psi)$ is winning for \plrsat, and since we assumed that $\sigma$ was a positional strategy it determines a choice $$\sigma(A_k(\bfi),\varphi \wedge \psi) \in \{(A_k(\bfi),\varphi), (A_k(\bfi),\psi)\}.$$
Suppose $\sigma(A_k(\bfi),\varphi \wedge \psi) = (A_k(\bfi),\varphi)$, since the other case is entirely analogous. We take $u_{k + 1}$ to be the corresponding premise of the application of the $\wedge$-rule, labelled $\Gamma,\at{i}\varphi$, and we set $A_{k + 1} = A_k$. If $\varphi$ is a literal then it must be in the closure of $\rho$ so since $\sigma$ is winning for \plrsat{} we have $M,A_k(\bfi) \nVdash \varphi$, hence $M_k,A_k(\bfi) \nVdash \varphi$ as required. If $\psi$ is not a literal, then we define a shadow play for the new trace in the obvious manner.

\paragraph{Rule $\mathsf{Glob}$:}
Easy, left to the reader. 

\paragraph{Modal rule:}
This is the most interesting case. 
Suppose that $u_k$ is the conclusion of an application of the modal rule of the form:
\begin{prooftree}
\AxiomC{$\Gamma,\at{i}\Box \varphi, \at{i}\Diamond \Psi,\at{j}\varphi,\at{j}\Psi$}
\UnaryInfC{$\Gamma, \at{i}\Box \varphi, \at{i}\Diamond \Psi$}
\end{prooftree}
Here we recall that $\bfj$ is a fresh variable, and by our assumption on the proof $\Pi$ no value has been assigned to $\bfj$ by $A_k$. Pick arbitrary traces leading to $\at{i}\Box \varphi$ and $\at{i}\Diamond \psi$ for $\psi \in \Psi$, respectively. By our assumption there are corresponding $\sigma$-guided shadow plays ending with positions $(A_k(\bfi),\Box \varphi$ and $(A_k(\bfi),\Diamond \psi)$ for $\psi \in \Psi$, respectively (since these formulas are not literals). At the position $(A_k(\bfi),\Box \varphi$, the positional winning strategy $\sigma$ for \plrsat{} picks some successor $w$ of $A_k(\bfi)$, and we set $A_{k + 1}(\bfj) = w$. To extend shadow plays associated with traces that end with non-literal formulas, the only interesting cases are traces of the form $\vec{t} \cdot (u_k,\at{i}\Box \varphi) \cdot (u_{k+1},\at{j}\varphi) $ and $\vec{t} \cdot (u_k,\at{i}\Diamond \psi) \cdot (u_{k+1},\at{j}\psi) $ where $\psi \in \Psi$ and $\vec{t} \cdot (u_k,\at{i}\Box \varphi)  $ and $\vec{t} \cdot (u_k,\at{i}\Diamond \psi) $ are traces leading to $u_k$. In both cases we have shadow plays $\pi \cdot (A_k(\bfi), \Box \varphi)$ and  $\pi \cdot (A_k(\bfi), \Diamond\psi)$ respectively. In the first case, we extend this to the unique play $\pi \cdot (A_k(\bfi), \Box \varphi) \cdot (w,\varphi)$ determined by  $\sigma$, which satisfies the required constraints by construction. In the second case, we extend  $\pi \cdot (A_k(\bfi), \Diamond\psi)$ to the $\sigma$-guided play $\pi \cdot (A_k(\bfi), \Diamond\psi) \cdot (w,\psi)$ by letting \plrval{} playing a move which is admissible since $w$ was a successor of $A_k(\bfi)$. Formulas of the form $\at{k}\psi$ where $\psi$ is a literal are handled as before.

To finish the proof, continuing the construction in this way we end up with an infinite branch $u_1u_2u_3...$ of $\Pi$, which must contain a good infinite trace since $\Pi$ is a valid proof. Note that such an infinite trace cannot contain any elements of the form $(u_k,\at{i}\psi)$ where $\psi$ is a literal, since there is no way such a trace can be continued to reach a new fixpoint unfolding.  It is clear that the shadow plays associated with initial segments of this trace are initial segments of an infinite $\sigma$-guided play in which the highest variable unfolded infinitely often is a $\nu$-variable. This play is thus a loss for \plrsat{}, and we have reached our contradiction. 
\end{proof}

\section{Finite proofs with names}

\subsection{The system \sysann}
\label{s:sysann}

In this section we introduce the finitary proof system \sysann, which is an annotated circular proof system in Stirling's style \cite{stir:tabl14}. We will use  a somewhat simplified version of the rules for manipulating annotations due to Afshari and Leigh \cite{afsh:cut17}.  For each fixpoint variable $x$ we assume that we have a countably infinite supply $\sfx_0,\sfx_1,\sfx_2...$ of \emph{names} for that variable. We assume that we have a fixed enumeration of the set of variable names for each variable $x$, so that we can speak of the $n$-th variable name for $x$. The system will be defined taking a strict linear order $<$ over fixpoint variables as a parameter, and in the presentation we assume such an order as given.  Given $<$ we write $\mathsf{x} < \mathsf{y}$ for names $\sfx, \sfy$ of variables $x,y$ respectively if $x < y$.  Given a word $\sfa$ over the set of variable names and a variable $x$, we write $\sfa \leq x$ if there is no variable $y > x$ for which $\sfa$ contains a name of $y$. Given two words $\sfa,\sfb$ over the set of variable names we write $\sfa \sqsubseteq \sfb$ to say that  $\sfb$ contains $\sfa$ as a subsequence. For example $\sfx \sfy \sqsubseteq \sfx \sfz \sfy$. We write $\sfa \sqcap \sfb$ for the longest word $\sfc$ such that $\sfc \sqsubseteq \sfa$ and $\sfc \sqsubseteq \sfb$, if a unique longest word with this property exists (otherwise $\sfa \sqcap \sfb$ is undefined). 
\begin{defi}
\emph{Annotated sequents} will be structures of the form:
$$\crl{a}\ats{i}{1}\varphi_1\ans{b}{1},...,\ats{i}{n}\varphi_n\ans{b}{n}$$
where $\sfa,\sfb_1,...,\sfb_n$ are non-repeating words over the set of variable names (i.e. no variable name appears twice in any of these words), each $\sfb_i$ is non-decreasing with respect to the order $<$, and $\sfb_i \sqsubseteq \sfa$ for each  $i \in \{1,...,n\}$. 
\end{defi} A formula $\rho$ will be said to be provable in the system if the sequent $\crm \at{r}\rho\an{\mpt}$ is provable, where the order $<$ on variable names is some arbitrary linearization of $<_\rho$,  where $\mpt$ is the empty word and $\mathsf{r}$ is a fresh nominal as before.  
\begin{defi}
A sequent in the sense of the system \sysinf{} will be called a \emph{plain sequent}. Given an annotated sequent $\Gamma = \crl{a}\ats{i}{1}\varphi_1\ans{b}{1},...,\ats{i}{n}\varphi_n\ans{b}{n}$, the \emph{underlying plain sequent} $\uls{\Gamma}$ is the plain sequent $\ats{i}{1}\varphi_1,...,\ats{i}{n}\varphi_n$.
\end{defi}

The system has two axioms, which are the law of exluded middle and an identity axiom, which now have the form: 
$$\crm\at{i}p\an{\mpt}, \at{i}\neg p\an{\mpt} \quad \quad \quad \crm \bfi \approx \bfi\an{\mpt}$$
Here, $p$ is a nominal or a propositional variable. Rules of inference are given in Figure \ref{fig:sysann-rules}. The rules are  subject to the following constraints:
\begin{description}
\item[$\mathsf{Mod}$:] The nominal $\bfj$ must be fresh. 
\item[$\eta x$:] $\sfb \leq x$. 
\item[$\mathsf{Rec}(\sfx)$:] $\sfb \leq x$, and $\sfx$ is a fresh variable name for $x$.
\item[$\mathsf{Exp}$:] $\sfa \sqsubseteq \sfa'$, $\sfb_i \sqsubseteq \sfb'_i$ and $\sfb_i'\sqcap \sfa \sqsubseteq \sfb_i$ for each $i$ \footnote{Note that $\sfb_i' \sqcap \sfa$ is well-defined here: since $\sfb_i' \sqsubseteq \sfa'$ and $\sfa \sqsubseteq \sfa'$, and since $\sfa'$ is non-repeating, any two variable names occurring in both $\sfb_i'$ and $\sfa$ must appear only once and in the same order in both words. From this follows that the set of words $\sfc$ such that $\sfc \sqsubseteq \sfb_i'$ and $ \sfc \sqsubseteq \sfa$ is a $\sqsubseteq$-directed finite set, so it contains a $\sqsubseteq$-maximal word.}. 
\item[$\mathsf{Reset}(\sfx)$:] The variable $x$ does not appear in any formula in $\Gamma$. 
\end{description}
\begin{figure}[h]
\fbox{
\begin{minipage}[t]{.30\textwidth}
\begin{prooftree}
\AxiomC{$\crl{a}\Gamma,  \at{i}\varphi \wedge \psi \an{b}, \at{i}\varphi\an{b}$}
\AxiomC{$\crl{a}\Gamma,  \at{i}\varphi \wedge \psi \an{b}, \at{i} \psi\an{b}$}
\RightLabel{$\wedge$}
\BinaryInfC{$\crl{a}\Gamma, \at{i}\varphi \wedge \psi \an{b}$}
\end{prooftree}
\vspace{.5cm}
\begin{prooftree}
\AxiomC{$\crl{a}\Gamma,\at{i} \phi\an{b}, \at{j}\phi\an{b}, \bfi \napprox \bfj\an{c}$}
\RightLabel{$\mathsf{Eq}$}
\UnaryInfC{$\crl{a}\Gamma, \at{j}\phi\an{b}, \bfi \napprox \bfj\an{c}$}
\end{prooftree}
\vspace{.5cm}
\begin{prooftree}
\AxiomC{$\crl{a}\Gamma, \at{i}\Box \varphi\an{b}, \at{i}\Diamond \Psi, \at{j}\varphi\an{b},\at{j}\Psi$}
\RightLabel{$\mathsf{Mod}$}
\UnaryInfC{$\crl{a}\Gamma, \at{i}\Box \varphi\an{b}, \at{i}\Diamond \Psi$}
\end{prooftree}
\end{minipage}
\begin{minipage}[t]{.33\textwidth}
\vspace{1.5cm}
\begin{prooftree}
\AxiomC{$\crl{a} \Gamma, \at{i}\eta x. \phi(x)\an{b}, \at{i}\phi(\eta x. \phi(x))\an{b}$} 
\RightLabel{$\eta x$}
\UnaryInfC{$\crl{a}\Gamma, \at{i}\eta x. \phi(x)\an{b}$}
\end{prooftree} 
\vspace{.1cm}
\begin{prooftree}
 \AxiomC{$\crl{a}\Gamma, \at{i}(\at{j}\varphi)\an{b}, \at{j}\varphi \an{b}$}
\RightLabel{$\mathsf{Glob}$}
\UnaryInfC{$\crl{a}\Gamma, \at{i}(\at{j}\varphi)\an{b}$}
\end{prooftree}
\vspace{1cm}
\begin{prooftree}
\AxiomC{$\crl{ax} \Gamma,  \at{i}\nu x. \varphi(x)\an{b}, \at{i}\varphi(\nu x. \varphi(x))\an{bx}$}
\RightLabel{$\mathsf{Rec}(\sfx)$}
\UnaryInfC{$\crl{a} \Gamma, \at{i}\nu x. \varphi(x)\an{b}$}
\end{prooftree}
\begin{prooftree}
\AxiomC{$\crl{a}\ats{i}{1}\varphi_1\an{b_1},....,\ats{i}{n}\varphi_n\an{b_n}$}
\RightLabel{$\mathsf{Exp}$}
\UnaryInfC{$\crl{a'}\ats{i}{1}\varphi_1\an{b_1'},....,\ats{i}{n}\varphi_n\an{b_n'}$}
\end{prooftree}
\begin{prooftree}
\AxiomC{$\crl{a}\Gamma, \at{i_1}\varphi_1\an{bx},....,\at{i_n}\varphi_n\an{bx}$}
\RightLabel{$\mathsf{Reset}(\sfx)$}
\UnaryInfC{$\crl{a}\Gamma,\at{i_1}\varphi_1\an{bxx_1c_1},....,\at{i_n}\varphi_n\an{bxx_nc_n}$}
\end{prooftree}
\vspace{.5cm}
\end{minipage}
\begin{minipage}[t]{.36\textwidth}
\begin{prooftree}
\AxiomC{$\crl{a}\Gamma, \at{i}\varphi \vee \psi \an{b},  \at{i}\varphi\an{b},\at{i} \psi\an{b}$}
\RightLabel{$\vee$}
\UnaryInfC{$\crl{a}\Gamma, \at{i}\varphi \vee \psi \an{b}$}
\end{prooftree}
\vspace{.5cm}
\begin{prooftree}
\AxiomC{$\crl{a}\Gamma, \bfi \napprox \bfj\an{c}, \bfj \napprox \bfi\an{c}$}
\RightLabel{$\mathsf{Com}$}
\UnaryInfC{$\crl{a}\Gamma, \bfi \napprox \bfj\an{c}$}
\end{prooftree}
\vspace{.5cm}
\begin{prooftree}
\AxiomC{$\crl{a}\Gamma$}
\RightLabel{$\mathsf{Weak}$}
\UnaryInfC{$\crl{a}\Gamma \cup \Psi$}
\end{prooftree}
\end{minipage}
}
\caption{Rules of \sysann}
\label{fig:sysann-rules}
\end{figure}

A \sysann-proof is a labelled tree where the labels specify a sequent assigned to a node and the last rule application (for non-leaf nodes),  and such that the children of a node are labelled with the premises of the specified rule application. Although valid proofs will always be finite it will be useful to consider infinite \sysann-proofs as well. Way say that the variable $x$ is \emph{reset} in an instance of the rule $\mathsf{Reset}(\sfx)$. We say that an infinite \sysann-proof is \emph{quasi-valid} if every infinite branch contains a good trace (defined as before) and every finite branch ends with an axiom. It is obvious that a quasi-valid infinite \sysann-proof can be turned into a valid \sysinf-proof, and we will not provide a detailed proof of this.


\begin{defi}
 A \sysann-proof will be considered  \emph{valid} if it is a finite proof-tree, and there is a map $f$ from non-axiom leaves to non-leaves, such that:
\begin{itemize}
\item $f(l)$ is an ancestor of $l$, and has the same label.
\item There is a variable name $\sfx$ that is contained in the control of every node in the path from $f(l)$ to $l$, and is reset at least once on this path. 
\end{itemize}
A map $f$ from non-axiom leaves to non-leaves is satisfying the first of these conditions is called a \emph{back-edge map}, and is \emph{good} if it satisfies the second condition as well.  So a finite proof-tree is a valid proof iff it has a good back-edge map.
\end{defi}

\begin{defi}
Let $\sfa$  be a name tuple and $x$ a variable. We write $\sfa \vert x$ for the result of removing all names of variables $y > x$ from $\sfa$, where we recall that $<$ was a linear order over variables taken as a parameter for the proof system. 
\end{defi}

\begin{defi}
Let $\sfa, \sfb,\sfc$ be three name tuples, where $\sfa$ contains both $\sfb,\sfc$. We write $\sfb <_\sfa \sfc$ if: 
either there is a $\nu$-variable $x$ for which  $\sfc\vert x$ is a proper prefix of $\sfb\vert x$, or there are  variable names $\sfz,\sfz'$ for the same variable and a name tuple $\sfd$ such that $\sfd\sfz$ is a prefix of $\sfb$, $\sfd\sfz'$ is a prefix of $\sfc$, and $\sfz$ is left of $\sfz'$ in $\sfa$. 
\end{defi}

The following derived rule of \sysann{} will be useful later:
\begin{prooftree}
\AxiomC{$\crl{a}\Gamma, \at{i}\varphi\an{b}$}
\RightLabel{$\mathsf{Thinning}$}
\UnaryInfC{$\crl{a'}\Gamma, \at{i}\varphi\an{b}, \at{i}\varphi\an{c}$}
\end{prooftree}
The rule is subject to the following constraints:
\begin{itemize} 
\item $\sfb <_{\sfa'} \sfc$.
\item $\sfa$ is obtained by removing all variables names  not appearing in $\Gamma, \at{i}\varphi\an{b}$ from $\sfa'$.
\end{itemize}

\subsection{Soundness}

In this subsection we prove soundness of the system \sysann.

\begin{defi}
Let $\Pi$ be a \sysann-proof and let $\pi$ be an infinite branch of $\Pi$. A variable name $\sfx$ is said to be an \emph{invariant} of $\pi$ if there exists some final segment of $\pi$ for which $\sfx$ belongs to the control of every node. We say that $\sfx$ is a \emph{good invariant} of $\pi$ if it is an invariant, and is reset infinitely many times on $\pi$. (That is, the rule $\mathsf{Reset}(\sfx)$ is applied infinitely many times on $\pi$.) 
\end{defi}

\begin{defi}
Notation: given a node $u$ in a \sysann-proof labelled $\crl{a}\Gamma$ and a variable name $\sfx$, we write $\mathsf{th}(\sfx,u)$ for the set:
$$\{\at{i}\varphi \mid \exists \sfb, \sfc:\;\at{i}\varphi\an{bxc} \in \Gamma\}$$
\end{defi}

\begin{prop}
\label{p:invariant-to-trace}
Let $\Pi$ be any infinite \sysann-proof. Then for any infinite branch $\pi$ of $\Pi$, if $\pi$ has a good invariant then $\pi$ contains a good infinite trace. 
\end{prop}

\begin{proof}
Fix a good invariant $\sfx$ of $\pi = (u_0,u_1,u_2...)$. Let a \emph{reset point} on $\pi$ be an index $i < \omega$ such that $u_i$ belongs to the final segment of $\pi$ in which $\sfx$ belongs to every control, and $u_i$ is the conclusion of an application of the reset rule in which $\sfx$ is reset. We enumerate the reset points on $\pi$ as $(r_0, r_1, r_2,...)$ so that $r_{i + 1}$ is a $\pi$-descendant of $r_i$ for each $i < \omega$.  The key claim is the following, the routine proof of which is omitted:
\begin{claim}
Let $i < \omega$. Then for every formula $\at{j}\varphi \in \mathsf{th}(\sfx, r_{i + 1})$ there exists a formula $\at{k}\psi \in \mathsf{th}(\sfx, r_{i})$ and a trace of the form $(r_i, \at{k}\psi) \cdot \vec{t} \cdot  (r_{i + 1}, \at{j}\varphi)$ on which $x$ is unfolded at least once. 
\end{claim}
To finish the proof, we note that since $\sfx$ belongs to the control of every node in a final segment of $\pi$, no higher-ranking variables than $x$ are ever unfolded in this final segment of $\pi$. So it suffices to show that there is an infinite trace in which $x$ is unfolded infinitely often. 

We construct a graph as follows: the nodes are pairs $(r,\at{j}\varphi)$ where $r$ is a reset point and $\at{j}\varphi \in \mathsf{th}(\sfx, r)$. We draw an edge between those pairs of nodes of the form $(r_i,\at{j}\varphi)$ and $(r_{i + 1}, \at{k}\psi)$ for which there exists a trace of the form $(r_i,\at{j}\varphi) \cdot \vec{t} \cdot (r_{i + 1}, \at{k}\psi)$ on which the variable $x$ is unfolded at least once. By the previous claim, this is an infinite connected graph, and it is clearly locally finite. So by Koenig's lemma it has an infinite simple path, and this path gives a good infinite trace on $\pi$.
\end{proof}

\begin{defi}
Let $\Pi$ be a \sysann-proof with back-edge map $f$. The  \emph{dependency} relation $D$ over the leaves of $\Pi$ is defined as follows: set $l D l'$ iff $f(l)$ is on the path from $f(l')$ to $l'$ in $\Pi$.  The  \emph{entanglement} relation $E$ over the leaves of $\Pi$ is defined as the symmetric closure of $D$. In other words, $l E l'$ iff the paths from $f(l)$ to $l$ and from $f(l')$ to $l'$ respectively intersect.  
\end{defi}

\begin{prop}
\label{p:intersecting-paths}
Let $\Pi$ be a \sysann-proof with back-edge map $f$, and let $l,l' \in \mathsf{dom}(f)$ be leaves such that $l E l'$. Suppose $\sfx, \sfy$ are variable names such that $\sfx$ is in the control of every node on the path from $f(l)$ to $l$ and $\sfy$ is in the control of every node on the path from $f(l')$ to $l'$. If $\sfx < \sfy$, then $\sfx$ is also in the control of every node on the path from $f(l')$ to $l'$. 
\end{prop}

\begin{proof}
Since $l E l'$ there is a node $u$ which is on both the path from $f(l)$ to $l$ and on the path from $f(l')$ to $l'$. So $\sfx$ is in the control of $u$. But since $\sfy$ is in the control of every node on the path from $f(l')$ to $l'$, and since $\sfx < \sfy$, the variable $x$ cannot be unfolded on any node between $f(l')$ and $l'$. It follows that $\sfx$ must already be in the control of $f(l')$. Furthermore, since $f(l')$ and $l'$ have the same control, and since again the variable name $\sfx$ cannot be introduced anywhere on the path from $f(l')$ to $l'$ by a variable unfolding, $\sfx$ must be in the control of every node in that path. 
\end{proof}

Since a $\sysann$-proof $\Pi$ with a back-edge map $f$ is a finite ranked tree with back-edges, the unfolding $\mathsf{unf}(\Pi,f)$ is a well-defined infinite \sysann-proof. 

\begin{prop}
\label{p:f-unfolding}
Let $\Pi$ be a finite \sysann-proof and $f$ a back-edge map for $\Pi$. Then $f$ is good iff every infinite path in $\mathsf{unf}(\Pi,f)$ has a good invariant.
\end{prop}

\begin{proof}
We prove each direction separately.

\textbf{Right to left:}
Suppose that every infinite path in $\mathsf{unf}(\Pi,f)$ has a good invariant. Let $l \in \mathsf{dom}(f)$. Then there is an infinite path in $\mathsf{unf}(\Pi,f)$ that simply repeats the path from $f(l)$ to $l$ in $\Pi$ forever, so there must be a good invariant $\sfx$ on this path. It is easy to see that the name $\sfx$ must be contained in the control of every node on the path from $f(l)$ to $l$ in $\Pi$, and that it must be reset at least once on this path.

\textbf{Left to right:}
Suppose that the back-edge map $f$ is good. Let $\pi$ be any infinite path in $\mathsf{unf}(\Pi,f)$, which we can identify with an infinite walk through $(\Pi,f)$, viewed is a directed graph obtained by adding back edges specified by $f$ to the tree $\Pi$.   We want to show that $\pi$ has a good trace.  Let $L$ be the set of leaves in $\Pi$ visited infinitely many times on $\pi$. Since $L \subseteq \mathsf{dom}(f)$, we can choose for each $l \in L$ the highest ranking variable name $\sfx$ that belongs to the control of every node on the path from $f(l)$ to $l$ and is reset at least once on this path, and call this variable name $\mathsf{var}(l)$.  It is easy to see that $L$ is a $D^*$-directed set, where $D^* $ is the transitive closure of the dependency relation $D$, hence $L$ is an $E$-connected set. If $\sfx$ is the highest ranking variable name in $\{\mathsf{var}(l) \mid l \in L\}$, it follows using Proposition \ref{p:intersecting-paths} that for every $l \in L$, $\sfx$ belongs to the control of every node in the path from $f(l)$ to $l$. Therefore, $\sfx$ belongs to the control of every node in a final segment of $\pi$. Furthermore, suppose $l \in L$ is such that $\sfx = \mathsf{var}(l)$. It is easy to show that between any two visits of the leaf $l$ on $\pi$, every edge on the path from $f(l)$ to $l$ must be traversed at least once. Hence $\sfx$ is reset infinitely many times on $\pi$.
\end{proof}

\begin{theo}[Soundness]
Any formula that has a valid \sysann-proof is semantically valid. 
\end{theo}
\begin{proof}
By Proposition \ref{p:f-unfolding}, Proposition \ref{p:invariant-to-trace} and Theorem \ref{t:infcompl}. 
\end{proof}

\subsection{Completeness}

In this section we prove the main result of the paper, completeness of the system \sysann. 
Our strategy is as follows: we begin with a frugal \sysinf-proof, which exists for every valid formula.  We show how to construct from this an infinite \sysann-proof in which only finitely many sequents appear, and on which every infinite branch has a good invariant. Next we note that if we fix a finite set of sequents, the set of infinite \sysann-proofs for a given formula, in which only those sequents may appear and in which every infinite branch has a good invariant, forms an MSO-definable tree language. So, since the tree language is non-empty, by Rabin's Basis Theorem we find a regular \sysann-proof still satisfying the good-invariant property. This regular \sysann-proof can then be turned into a valid finite proof.

\begin{defi}
Let $\Pi$ be an \sysinf-proof. A \emph{decoration} of $\Pi$ is an assigment $d$ to each node $u$ in $\Pi$ of a finite \sysann-proof such that:
\begin{itemize}
\item For each node $u$, the underlying plain sequent of the label of the root of $d(u)$ is equal to the label of $u$.
\item If $u$ is not a leaf, then there is a bijective correspondence $i$ from leaves of $d(u)$ to premises of $u$ such that for each leaf $l$ in $d(u)$, the root of $d(i(l))$ has the same label as $l$. 
\end{itemize}
Given a decoration $d$ and $u \in \Pi$ we let $\overline{d(u)}$ denote the witnessing bijective correspondence from leaves of $d(u)$ to premises of $u$.
\end{defi}

Decorations can be used to turn \sysinf-proofs into infinite \sysann-proofs according to the following coinductive definition:

\begin{defi}
Let $\Pi$ be an \sysinf-proof and $d$ a decoration. Then $(\Pi[d],\widehat{d})$ is the unique pair in which $\Pi[d]$ is a $\sysann$-proof and $\widehat{d}$ is a map from $\Pi$ to $\Pi[d]$ such that, for each $u \in \Pi$: 
$$\Pi[d]\rst{\widehat{d}(u)} = d(u)[\Pi[d]\rst{\overline{d(u)}(l)} / l \mid l \text{ a leaf of } d(u)]$$ 
\end{defi}

\begin{defi}
Let $\Pi$ be an \sysinf-proof. We define the \emph{canonical decoration} $d(u)$ of a node $u$ in $\Pi$ by induction on the height of the node $u$ as follows. Suppose that $u$ is some node and the decoration $d$ has been defined for all nodes of lower height. 
The construction of $d(u)$ will be  carried out in a number of steps, of which the most interesting ones essentially mimick the update procedure for Safra trees in determinization of stream automata \cite{safr:comp88}. The construction of each bijection $\overline{d(u)}$ will be quite obvious so we omit it. We shall maintain the invariant that, for each $u$ and each leaf $v$ in $d(u)$, in the label of $v$ no formula appears with more than one annotation.

\paragraph{Step 1: find the root label.}
First we define the label of the root of $d(u)$ as follows: if $u$ is the root of $\Pi$ then we label the root of $d(u)$ by the unique sequent for which all annotations and the control are empty, and for which the underlying sequent is the label of $u$ in $\Pi$. If $u$ is not the root then let $v$ be its parent node. We set the label of the root of $d(u)$ to be equal to that of the leaf $\overline{d(v)}^{-1}(u)$ in $d(v)$.

\paragraph{Step 2: dealing with leaves.}
If $u$ is a leaf labelled with an axiom then $d(u)$ is constructed by applying Expansion to the root label determined by Step 1, so that both control and all annotations are empty in the unique lead of $d(u)$. If $u$ is not a leaf then we skip this step. 

\paragraph{Step 3: register fixpoint unfoldings.}

If $u$ is not the conclusion of an application of the $\nu$-rule or $\mu$-rule then we skip this step. Otherwise, suppose that $u$ is the conclusion of a rule application of the following shape:
\begin{prooftree}
\AxiomC{$\Gamma,\at{i}\nu x. \varphi(x),\at{i}\varphi(\nu x. \varphi(x))$}
\UnaryInfC{$\Gamma,\at{i}\nu x. \varphi(x)$}
\end{prooftree}
We focus on the case of a greatest fixpoint unfolding since the other case is simpler.  
Suppose the root label of $d(u)$ was determined in Step 1 to be $\crl{a}\Theta,\at{i}\nu x. \varphi(x)\an{b}$ where $\underline{\Theta} = \Gamma$. The current stage in the construction of $d(u)$  is then shown below:
\begin{prooftree}
\AxiomC{$(\sfa \vert x)\sfx \vdash \{\at{j}\theta^{\sfa\vert x} \mid \at{j}\theta\an{a} \in \Theta\}, \at{i}\nu x. \varphi(x)^{\sfb\vert x},\at{i}\varphi(\nu x. \varphi(x))^{\sfb\vert x \sfx}$}
\RightLabel{$\mathsf{Rec}(\sfx)$}
\UnaryInfC{$\sfa \vert x \vdash  \{\at{j}\theta^{\sfa\vert x} \mid \at{j}\theta\an{a} \in \Theta\} ,\at{i}\nu x. \varphi(x)^{\sfb\vert x},\at{i}\nu x. \varphi(x)^{\sfb\vert x}$}
\RightLabel{$\mathsf{Exp}$}
\UnaryInfC{$\crl{a} \Theta,\at{i}\nu x. \varphi(x)\an{b}$}
\end{prooftree}
Here, $\mathsf{x}$ is the smallest fresh name for the variable $x$ in the fixed enumeration of the variable names. 

\paragraph{Step 4: apply other rules.}

If step 1 was applied then we skip this step. Otherwise, we make a case distinction as to which rule is applied to $u$. If $u$ is the conclusion of an instance of Weakening, then we apply Weakening to the root label of $d(u)$ to remove the annotated version of each formula removed from the label of $u$, followed by an application of Expansion to remove any superfluous variable names in the control that no longer appear in annotations of any formulas.  If the rule applied was the $\wedge$-rule, $\vee$-rule, $\mathsf{Mod}$, $\mathsf{Eq}$, $\mathsf{Com}$ or $\mathsf{Glob}$, then we apply the corresponding rule instance to the root label of $d(u)$ determined in Step 1, recalling that none of these rules will affect the annotations.  

\paragraph{Step 5: merge traces.}

If possible, repeatedly apply Thinning to each leaf in the proof-tree produced by Steps 1 -- 4 until no further applications of Thinning are possible. 

\paragraph{Step 6: reset.}

If possible, repeatedly apply the Reset rule to each leaf in the proof-tree produced by Step 5 until no further applications are possible. 
\end{defi}

As an example showing how the canonical decoration works, consider the following part of an \sysinf-proof $\Pi$: 
\begin{prooftree}
\AxiomC{$\vdots$}
\noLine
\UnaryInfC{$\at{j}\varphi, \at{j}\Box \varphi, \at{j}\psi$}
\RightLabel{$\vee$}
\UnaryInfC{$ \at{j}(\varphi \vee \Box \varphi), \at{j}\psi$}
\AxiomC{$\vdots$}
\noLine
\UnaryInfC{$\at{k}(\varphi \vee \Box \varphi)$}
\RightLabel{$\mathsf{Mod}$}
\UnaryInfC{$\at{j}\Box(\varphi \vee \Box \varphi), \at{j}\Box \varphi$}
\RightLabel{$\nu x$}
\UnaryInfC{$\at{j}\varphi, \at{j}\Box \varphi$}
\RightLabel{$\vee$}
\UnaryInfC{$\at{j}(\varphi \vee \Box \varphi), \at{j}\varphi$}
\RightLabel{$\wedge$}
\BinaryInfC{$\at{j}(\varphi \vee \Box \varphi), \at{j}(\psi \wedge \varphi)$}
\RightLabel{$\mathsf{Mod}$}
\UnaryInfC{$\at{i}\Box(\varphi \vee \Box \varphi), \at{i}\Diamond(\psi \wedge \varphi)$}
\RightLabel{$\nu y$}
\UnaryInfC{$\at{i}\Box(\varphi \vee \Box \varphi),  \at{i}\psi$}
\RightLabel{$\nu x$}
\UnaryInfC{$\at{i}\varphi, \at{i}\psi$}
\end{prooftree}
Here, $\varphi$ is the formula $\nu x.\Box(x \vee \Box x$ and $\psi$ is the formula $\nu y.\Diamond(y \wedge\varphi)$. We assume the order $<$ is chosen so that $x < y$. For readability, we have suppressed applications of Weakening in this proof, which only serve to remove principal formulas of rule applications in the proof.  
The corresponding part of the infinite \sysann-proof $\Pi[d]$ read off from the canonical decoration $d$ is shown below. 
\begin{prooftree}
\AxiomC{$\vdots$}
\noLine
\UnaryInfC{$\sfx_0\sfy_0 \vdash \at{j}\varphi\ans{x}{0}, \at{j}\Box \varphi\ans{x}{0}, \at{j}\psi\ans{y}{0}$}
\RightLabel{$\vee$}
\UnaryInfC{$\sfx_0\sfy_0 \vdash \at{j}(\varphi \vee \Box \varphi)\ans{x}{0}, \at{j}\psi\ans{y}{0}$}
\AxiomC{$\vdots$}
\noLine
\UnaryInfC{$\sfx_0 \vdash \at{k}(\varphi \vee \Box \varphi)^{\sfx_0}$}
\RightLabel{$\mathsf{Reset}(\sfx_0)$}
\UnaryInfC{$\sfx_0\sfx_1 \vdash\at{k}(\varphi \vee \Box \varphi)^{\sfx_0\sfx_1}$}
\RightLabel{$\mathsf{Mod}$}
\UnaryInfC{$\sfx_0\sfx_1 \vdash \at{j}\Box(\varphi \vee \Box \varphi)^{\sfx_0\sfx_1}, \at{j}\Box \varphi\ans{x}{0}$}
\RightLabel{$\mathsf{Rec}(\sfx_1)$}
\UnaryInfC{$\sfx_0 \vdash \at{j}\varphi\ans{x}{0}, \at{j}\Box \varphi\ans{x}{0}$}
\RightLabel{$\mathsf{Thinning}$}
\UnaryInfC{$\sfx_0\sfy_0 \vdash \at{j}\varphi\ans{x}{0}, \at{j}\Box \varphi\ans{x}{0}, \at{j}\varphi\ans{y}{0}$}
\RightLabel{$\vee$}
\UnaryInfC{$\sfx_0\sfy_0 \vdash \at{j}(\varphi \vee \Box \varphi)\ans{x}{0}, \at{j}\varphi\ans{y}{0}$}
\RightLabel{$\wedge$}
\BinaryInfC{$\sfx_0\sfy_0 \vdash \at{j}(\varphi \vee \Box \varphi)\ans{x}{0}, \at{j}(\psi \wedge \varphi)\ans{y}{0}$}
\RightLabel{$\mathsf{Mod}$}
\UnaryInfC{$\sfx_0\sfy_0 \vdash \at{i}\Box(\varphi \vee \Box \varphi)\ans{x}{0}, \at{i}\Diamond(\psi \wedge \varphi)\ans{y}{0}$}
\RightLabel{$\mathsf{Rec}(\sfy_0)$}
\UnaryInfC{$\sfx_0\vdash \at{i}\Box(\varphi \vee \Box \varphi)\ans{x}{0},  \at{i}\psi\an{\mpt}$}
\RightLabel{$\mathsf{Rec}(\sfx_0)$}
\UnaryInfC{$\crm\at{i}\varphi\an{\mpt}, \at{i}\psi\an{\mpt}$}
\end{prooftree}
Again, some applications of Weakening have been hidden.

It remains to be checked that the construction really does maintain the invariant that each formula in the label of a leaf of $d(u)$ for some node $u$ appears with at most one annotation. Since the canonical decoration has been set up so that Thinning will be applied whenever possible, we only need to show that Thinning must apply to any sequent in which some formula appears with more than one annotation. We need to check that if two annotations $\at{i}\varphi\an{a},\at{i}\varphi\an{b}$ of the same formula $\at{i}\varphi$ appear in some sequent with control $\sfc$, then $\sfa <_\sfc \sfb$ or $\sfb <_\sfc \sfa$ (but not both). The proof is the same as for Lemma 4.31 in Jungteerapanich's thesis \cite{jung:tabl10}. We therefore omit the details and refer the interested reader to that publication. 

Since the Reset rule is applied whenever possible (assuming that Thinning does not apply), one can also show that the length of annotations and control appearing in a sequent is bounded. The proof is the same as for Lemma 4.39 in \cite{jung:tabl10}. Again we refer the reader to that publication for details. Since only finitely many fixpoint variables can appear in a proof, and since names introduced by the $\mathsf{Rec}$-rule are always chosen canonically to be the fresh name with lowest index, we get the following result.   

\begin{prop}
\label{p:thin-works}
Let $\Pi$ be a frugal $\sysinf$-proof and let $d$ be its canonical decoration. Then $\Pi[d]$ is also frugal, i.e. only finitely many annotated sequents appear in $\Pi[d]$.
\end{prop}

The following proposition is a ``soundness'' result for the canonical decoration, showing that it successfully detects good traces on infinite branches. 

\begin{prop}
\label{p:can-good}
Let $\Pi$ be a valid, frugal \sysinf-proof and $d$ its canonical decoration. Then every infinite branch of $\Pi[d]$ has a good invariant. 
\end{prop}

\begin{proof}
Suppose that $\tau$ is a good infinite trace in an infinite branch $\pi$.  Let $x$ be the highest ranking ($\nu$-)variable that is unfolded infinitely many times on $\tau$. Say that a word $\sfa$ over the set of all variable names is $\tau$-\emph{stable} if there is a final segment of $\tau$ in which every formula has an annotation $\sfb$ with $\sfa \sqsubseteq \sfb$. Let $S$ be the set of $\tau$-stable name tuples; clearly $S$ is non-empty since $\varepsilon \in S$, and it is finite since there are only finitely many annotations in $\Pi[d]$. Pick a maximal element of $S$, i.e. a $\tau$-stable tuple $\sfa$ such that no  $\sfb \in S$ exists with $\sfa \sqsubset \sfb$. We want to show that $\sfa$ is a non-empty tuple and that its last variable name is reset infinitely many times on $\pi$. Since annotations are always contained in the control, it follows that the last name  on $\sfa$ is a good invariant of $\pi$. 

We first check that the tuple $\sfa$ is non-empty - this amounts to showing that at least one variable name belongs to every annotation in some final segment of $\tau$. Clearly the annotation in $\tau$ will be non-empty from some point onwards, and by inspection of the update procedure we see that its left-most element can only change by being replaced by a name to the left of it in the control. So eventually the left-most element of all annotations will stay the same, since such moves to the left can only happen finitely many times.

Given that $\sfa$ is non-empty, we now prove the required result by contraposition. Suppose the last variable name in $\sfa$ is reset at most finitely many times on $\pi$. We show  $\sfa$ is not maximal in $S$.  
Let $\pi'$ be a final segment of $\pi$ such that:
\begin{itemize}
\item  $\sfa$ is contained in every annotation on the $\pi'$-part of $\tau$,
\item the last variable name of $\sfa$ is \emph{never} reset in $\pi'$.
\end{itemize}
By inspection of the update procedure for annotations, it is easy to see that all annotations on the $\pi'$-part of $\tau$ have the same initial segment up to and including the last name in $\sfa$. Since $\sfa$ is contained in every annotation on the $\pi'$-part of $\tau$, it follows that no higher ranking variable than the one named by the names in $\sfa$ is ever unfolded in the $\pi'$-part of $\tau$.

By the construction of the decoration $d$, since $x$ is unfolded infinitely many times on $\tau$,  the $\pi'$-part of the trace $\tau$ will eventually reach a step of the form $(u,\at{i}\nu x.\varphi(x))\an{b} \cdot (v,\at{i}\varphi(\nu x.\varphi(x))\an{bx})$ where $\sfx$ is a fresh variable name for $x$ not appearing in $\sfb$. Since $\sfa \sqsubseteq \sfb$, we have $\sfa \sfx \sqsubseteq \sfb \sfx$. If the variable name $\sfx$ is never removed later in the trace $\tau$ then $\sfa \sfx$ is a proper extension of $\sfa$ that is $\tau$-stable. So suppose it is removed at some point.  Since $\sfx$ is a name for $x$, it can never be removed by an unfolding of a higher-ranking variable in the $\pi'$-part of $\tau$. So this removal can only happen in two ways: an application of the Reset rule or a trace merge (an application of Thinning). We show that in each case, we can find a name $\sfz$ in the new annotation that is to the left of $\sfx$ in the control, but to the right of all names in $\sfa$.  Since new variable names are always appended on the right side of the control, such ``moves to the left'' in the control can only happen finitely many times. So eventually we find a stable proper extension of $\sfa$. 

If $\sfx$ was removed due to an application of Reset, then the desired conclusion follows immediately since the last variable on $\sfa$ is never reset on the $\pi'$-part of $\tau$, and since $\sfa$ is never removed from the annotations.  The variable name being reset must therefore be strictly between $\sfx$ and $\sfa$ in the annotation, and hence to the left of $\sfx$ in the control. 

In case of an application of Thinning, some annotation $\sfd$ containing $\sfx$ is replaced on the trace $\tau$ by another annotation $\sfd'$ that does not contain $\sfx$, and such that $\sfd' <_\sfc \sfd$ where $\sfc$ is the control of the conclusion of the rule application. There are two cases to consider here: either there is a $\nu$-variable $y$ for which  $\sfd\vert y$ is a proper prefix of $\sfd'\vert y$, or there are  variable names $\sfy \neq \sfy'$ for the same variable $y$ and a name tuple $\sfd''$ such that $\sfd''\sfy'$ is a prefix of $\sfd'$, $\sfd''\sfy$ is a prefix of $\sfd$, and $\sfy'$ is left of $\sfy$ in the control $\sfc$. 

In the first case,  since the annotations $\sfd$ and $\sfd'$ are the same up to and including the last variable name in $\sfa$, and $\sfd\vert y$ is a \emph{proper} prefix of $\sfd'\vert y$, all names of the variable $y$ in $\sfd'$ must be to the right of all names in $\sfa$. So $y$ is either equal to $x$ or lower-ranking. In both cases, the name $\sfx$ is part of the prefix $\sfd\vert y$ and hence stays in $\sfd'$.

In the second case, note that $\sfy'$ cannot be in the common part of the annotations up to and including the last variable name in  $\sfa$ since $\sfy \neq \sfy'$. So $\sfy'$ is to the right of all names in  $\sfa$ in the annotation  $\sfd'$. Now, consider two possible places where $\sfx$ might appear in $\sfc$: if it is to the left of $\sfy'$ then it must be in the common prefix $\sfd''$ and thus it is still in the annotation $\sfd'$. It cannot be equal to $\sfy'$ since we assumed it does not belong to $\sfd'$. If it is to the right of $\sfy'$ in the annotation $\sfd'$, then $\sfy'$ is to the left of $\sfx$ in the control. So we have found a variable name that does not belong to $\sfa$ by moving to the left in the control, as required. Hence, the proof is finished.  
\end{proof}

\begin{prop}
\label{p:regularization}
Given any infinite \sysann-proof in which at most finitely many sequents appear, every leaf is an axiom and each infinite branch has a good invariant, there exists an infinite but regular \sysann-proof with the same root formula and satisfying the same criteria.  
\end{prop}
\begin{proof}
Easy consequence of Rabin's Basis Theorem, by noting that the required conditions on a proof tree with a fixed set of annotated sequents appearing as labels are MSO-definable. 
\end{proof}

We can now prove the main result.

\begin{theo}[Completeness]
Any semantically valid formula has a valid \sysann-proof. 
\end{theo}
\begin{proof}
By Theorem \ref{t:infcompl} any valid formula has a valid and frugal \sysinf-proof $\Pi$, and by Proposition \ref{p:can-good} every infinite branch of $\Pi[d]$ has a good invariant, and $\Pi[d]$ has only finitely many sequents as labels by Proposition \ref{p:thin-works}. By Proposition \ref{p:regularization} we may assume that $\Pi[d]$ is regular. By Proposition \ref{p:folk}, there is a finite \sysann-proof $(\Pi',f)$ the unfolding of which is isomorphic with $\Pi[d]$. By Proposition \ref{p:f-unfolding}, $\Pi'$ is a valid \sysann-proof. 
\end{proof}

\section{Concluding remarks}

We conclude with some directions for future work. First of all, with the Stirling-style proof system in place for the hybrid $\mu$-calculus, we should be able to prove cut-free completeness of a sequent system for the hybrid $\mu$-calculus by following the same method of translation between proof systems as in \cite{afsh:cut17}. The proof should not involve any substantial novelties, although the details remain to be checked. 

But perhaps more interesting, we hope that the methods developed here can be extended to other extended $\mu$-calculi. A first task in this direction is to consider converse modalities, and obtain a Stirling-style circular system for the hybrid $\mu$-calculus including converse modalities. In Vardi's automata-theoretic decision procedure for the two-way $\mu$-calculus, the key component is a finite data structure for encoding generalized traces that can go upwards or downwards along branches in a tableau. This extra component is then removed through a projection operation on automata recognizing valid tableaux. It would be interesting to investigate this construction from a proof-theoretic perspective. 

Further down the line, we hope that proof theory for even more expressive systems could be developed, like guarded fixpoint logic. The guiding intuition which is to be tested is the following: if a fixpoint logic has an automata-theoretic decision procedure for satisfiability in which it the logic is treated ``as if it had the tree-model property'', then a sound and complete circular proof system can be given, and ultimately a complete standard sequent calculus.  The present work can be seen as corroborating this hypothesis for the case of the hybrid $\mu$-calculus. 
 
\paragraph{Acknowledgement}
This research was supported by the Swedish Research Council grant 2015-01774.

{
\bibliographystyle{plain}
\bibliography{references}
}

\end{document}